\documentclass[amsmath,amssymb,amsfonts,aps,prl,superscriptaddress,bibnotes,showpacs,showkeys,longbibliography,notitlepage,10pt,twocolumn]{revtex4-1}

\usepackage{mathtools}
\usepackage{epsfig}
\usepackage[english]{babel}
\usepackage{color}
\usepackage{subcaption}
\usepackage{hyperref}
\usepackage{lipsum}
\usepackage{xcolor}
\usepackage{amsthm}
\usepackage{booktabs}
\usepackage{multirow}

\newtheorem{theorem}{Theorem}

\AtBeginDocument{%
    \newwrite\bibnotes
    \def\bibnotesext{Notes.bib}
    \immediate\openout\bibnotes=\jobname\bibnotesext
    \immediate\write\bibnotes{@CONTROL{REVTEX41Control}}
    \immediate\write\bibnotes{@CONTROL{%
    apsrev41Control,author="08",editor="1",pages="1",title="0",year="1"}}
     \if@filesw
     \immediate\write\@auxout{\string\citation{apsrev41Control}}%
    \fi
}%

\newcommand{\me}{\mathrm{e}}

\begin{document}
\title{Hierarchical deposition and scale-free networks: a visibility algorithm approach}

\author{Jonas Berx}
\affiliation{Institute for  Theoretical Physics, KU Leuven, B-3001 Leuven, Belgium}

\date{\today}

\begin{abstract}
The growth of an interface formed by the hierarchical deposition of particles of unequal size is studied in the framework of a dynamical network generated by a horizontal visibility algorithm. For a deterministic model of the deposition process, the resulting network is scale-free with dominant degree exponent $\gamma_e = \ln{3}/\ln{2}$ and transient exponent $\gamma_o = 1$. An exact calculation of the network diameter and clustering coefficient reveals that the network is scale invariant and inherits the modular hierarchical nature of the deposition process. For the random process, the network remains scale free, where the degree exponent asymptotically converges to $\gamma =3$, independent of the system parameters. This result shows that the model is in the class of fractional Gaussian noise (fGn) through the relation between the degree exponent and the series' Hurst exponent $H$. Finally, we show through the degree-dependent clustering coefficient $C(k)$ that the modularity remains present in the system.
\end{abstract}

\maketitle

\section{Introduction}\label{sec:intro}
The deposition of particles onto a substrate is a subject of considerable practical importance, and it has found applications in various fields \cite{barabasi_stanley_1995}. A plethora of deposition processes exists within the area of statistical physics, which often reduce to variations of either the \emph{random} or the \emph{ballistic} deposition process \cite{Family_1985,Family_1986,Meakin1986}. These surface growth models are understood to be the microscopic discrete versions of the continuous models of respectively Edwards and Wilkinson \cite{Edwards1982,Vvedensky1993,Buceta2014} (if surface relaxation is taken into account), and Kardar, Parisi and Zhang \cite{KPZ,Bertini1997}. Through experimental and theoretical studies performed over the past decades, it is known that such surfaces are rough and generally display fractal behaviour \cite{Gomes2021}. One assumption that is often made in the discrete variations is that the particles are geometrically identical in every aspect, and that the deposition steps occur \emph{sequentially}, i.e. one-by-one. These assumptions are not generally true for real-world applications. Here, we will consider the hierarchical deposition model (HDM) that was introduced in Ref. \cite{INDEKEU1998294}, which assumes that particles are deposited according to a power law as a function of their size in a \emph{synchronous} fashion, where particles of the same size are all deposited simultaneously in ``generations", analogous to a finite density aggregation process where the incoming particle flux can be controlled \cite{Cheng1987,Baiod1988,Krug1991}. Similar studies on systems with power-law particle distributions investigated sequential deposition by including a power-law distributed noise, resulting in rare-event dominated fluctuations \cite{ZHANG19901,Hosseinabadi2019}. Natural processes that show a power-law size distribution of the particles can range from the large-scale, e.g. the deposition of pyroclastic material from volcanic eruptions \cite{Pioli2019}, or the mass distribution of fragments of an asteroid breaking up in the Earth’s atmosphere \cite{Brykina2021}, to small laboratory scales such as the debris of nuclei multifragmentation \cite{Finn1982}. Generally, a panoply of solid materials (rock, sand, concrete, etc.), possess a power-law size distribution resulting from the breaking of larger objects \cite{Turcotte1986}. The HDM has been extended to include different scaling of particle size and deposition probability \cite{POSAZHENNIKOVA2000,Indekeu2001,INDEKEU2000135}, and has been used to model the deposition of particles with spins \cite{magneticDeposition}. Furthermore, it has found applications in the description of bacterial biofilms \cite{INDEKEU200414,INDEKEUSZNAJD} and has been used to describe coastline formation and fractal percolation \cite{BERX2021125998}.

In this work, we map the topology of the surface formed by the HDM into a complex network by means of a \emph{visibility algorithm} (VA). The essence of the VA is to create a network from a set of data by assigning a node to each datum and assigning edges based on the mutual visibility between two data, i.e., if a line of visibility is not ``intersected" by any intermediate data. This algorithm was originally developed \cite{Lacasa4972} to uncover structures in time series data, uch as finding signatures of self-organised criticality (SOC) in avalanche-based data \cite{Kaki2022}, and it has found applications in astrophysics \cite{npg-19-657-2012}, medicine \cite{Ahmadlou20101099}, fluid mechanics \cite{Juniper2018661} and several other fields \cite{Paiva2022}. Generally, the VA comes in two types: the Natural Visibility Algorithm (NVA) and the Horizontal Visibility Algorithm (HVA). The former considers the visibility line as a direct connection between two data and hence it is inclined with respect to the chosen common baseline. The latter is based on whether or not two data can have mutual horizontal visibility (illustrated in the middle panel of Fig. \ref{fig:diameter_HVG}) and has no inclination with respect to the baseline. It is not difficult to see that the graph generated by the HVA is always a subgraph of the NVA, as horizontal visibility implies natural visibility between two data. An extensive comparison between the two choices is performed in Ref. \cite{Lacasa4972}. The visibility algorithm has been extended to include, e.g., directed networks, which can be used to determine time series (ir)reversibility \cite{Lacasa2012}, limited penetrable (i.e., see-through) visibility networks \cite{Wang2018}, and multilayer networks to study multivariate data \cite{Lacasa2015}. A comprehensive review of VAs can be found in Ref. \cite{Nunez12}.

The setup of this work is as follows. In section \ref{sec:HDD}, we introduce a deterministic version of the hierarchical deposition model and calculate various topological quantities: the degree distribution, graph diameter, clustering coefficient and adjacency matrix spectrum. In section \ref{sec:HRD}, we extend the model to incorporate random deposition and erosion with probabilities $P$ and $Q$, respectively. We revisit the topological quantities and calculate them numerically. We analytically argue that the degree exponent is $\gamma = 3$ through a connection with the Hurst exponent $H$ and the Hausdorff dimension $D_f$. Finally, in section \ref{sec:conclusions}, we present conclusions and a future outlook.

\section{Hierarchical deterministic deposition}\label{sec:HDD}
Before studying the random version of the HDM, let us start by introducing the hierarchical deterministic deposition model (HDDM). This model will serve to illustrate the main concepts in an analytical and mechanistic fashion. Consider the deposition of rigid particles on a one-dimensional substrate, where we assume that deposited particles are squares with different sizes. If the deposition occurs in a viscous medium such as water, the larger particles are deposited first, followed by the smaller particles. Additionally, we assume that the number of particles $N(s)$ of size $s$ follows a hyperbolic distribution 
\begin{equation}
    \label{eq:hyperbolic}
    N(s) = \lambda^{-1}N(s/\lambda)\,,
\end{equation}
where $\lambda\in\mathbb{N}$ and $\lambda > 1$. Particles are deposited simultaneously instead of sequentially, in contrast with, e.g., the solid-on-solid or the ballistic deposition models. For simplicity, let us only consider $\lambda = 3$ and divide the unit interval $[0,1]$ into segments of length $1/3$. Deposit a square with side length $1/3$ on the middle segment while leaving the others empty. Repeating this procedure of segmentation and subsequent deposition of particles results in a ``castle'' landscape as shown in Fig. \ref{fig:diameter_HVG}(a), where one step of segmentation and deposition corresponds to a single generation $n$.
\begin{figure*}[htp]
     \centering
     \includegraphics[width=0.95\linewidth]{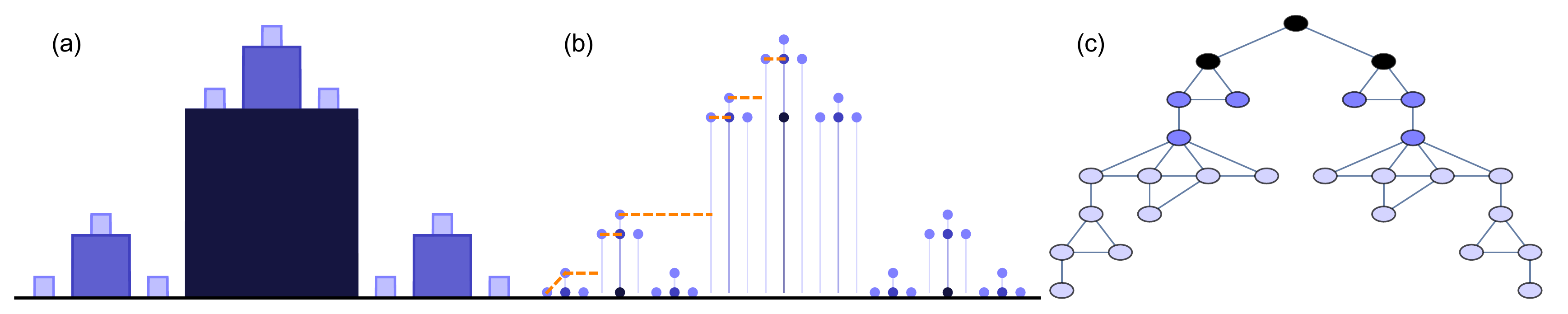}
        \caption{\textbf{(a)} Iterative construction of the fractal landscape for the first three generations of the HDDM. \textbf{(b)} Height profile for the fractal landscape. The shortest path length between the leftmost and central nodes is indicated by dashed orange lines. \textbf{(c)} The complex network based on the HVA as defined by equation \eqref{eq:visibility}. The different stages of each figure are indicated by colors: generation $n=1$ (black), $n=2$ (dark blue), and $n=3$ (light blue)}.
        \label{fig:diameter_HVG}
\end{figure*}

In Ref. \cite{INDEKEU1998294}, it was shown that this construction results in a \textit{logarithmic} fractal, whereby the surface length increment quickly saturates to the constant $2/3$. This logarithmic fractal character remains present for a random model with deposition probability $P$. 

From the above construction, it is clear that not every value for the interface height, measured with respect to the baseline, is accessible. In particular, the set of height values (in ascending order) in a generation $n$ is the subset of the first $2^n$ values of the set $3^{-n} S(0,1)$, i.e., $H^{(n)} = 3^{-n} S^{(n)}(0,1) \subset 3^{-n}S(0,1)$ where $S(0,1) = \{0, 1, 3, 4, 9, 10, 12, 13,\dots\}$ is the \emph{Stanley sequence} (OEIS A005836). The height distribution function, supported on $H^{(n)}$, is then 
\begin{equation}
    \label{eq:height_distribution}
    P^{(n)}(y = h_i\in H^{(n)}) = \left(\frac{2}{3}\right)^n 2^{-s(i)}\,,
\end{equation}
where $i$ is the index of $h$ in the set $H^{(n)}$, starting from $i=0$, and the sequence $s(i)$ indicates the number of $1$'s in a binary expansion of $i$, defined through the recurrence 
\begin{align}
\label{eq:recurrence}
        s(0)&=0\,, && s(2i)=s(i)\,, && s(2i+1)= s(i)+1\,.
\end{align}
By construction, the height at the central point of the interface grows to the maximal value $h_{max} = \sum_{i=1}^n 3^{-i}$, which is bounded by $1/2$ for large $n$.


We now show that the HDDM can be mapped to a scale-free network, the horizontal visibility graph (HVG), by means of the HVA \cite{Lacasa4972}. 

Two data points $x_i$ and $x_j$ possess mutual horizontal visibility if the following criterion holds:
\begin{align}
\label{eq:visibility}
    \min{\{x_i\,,x_j\}}>x_k && \text{for all $k$}\quad\text{such that} \,i<k<j\,.
\end{align}
If two points $x_i,x_j$ have mutual visibility, an edge exists between the two nodes $i,j$ in the network created by the HVA. This algorithm has been used to characterise time series \cite{Luque2009} by converting them to complex networks, enabling one to study the series combinatorically. To the author's knowledge, VAs have not received much attention in the context of deposition processes, with only one work \cite{Bru2014} that deduces growth exponents for the Edwards-Wilkinson (EW), Kardar-Parisi-Zhang (KPZ) and Molecular Beam Epitaxy (MBE) equations, as well as for the discrete Random Deposition (RD), Random Deposition with Surface Relaxation (RDSR) and Eden models. Another related work \cite{KARTHA2017556} studies the random deposition of patchy particles in the context of VAs. We will now study different topological aspects of the network associated with the HDDM.

\subsection{The degree distribution}
When the HVA is applied to the HDDM, the node degree $k_j(n)$ can be found for every individual node $j = 1,...,3^n$ in a generation $n$ by summing all entries in the $j-$th column of the adjacency matrix $A^{(n)}$, i.e., $k_j(n) = \sum_{i=1}^{3^n}A^{(n)}_{ij}$. After some algebra, which is outlined in appendix \ref{app:A}, and by making use of the hierarchical structure of $A^{(n)}$, the $j$th node degree can be found as follows
\begin{widetext}%
\begin{equation}%
\label{eq:degree}
    \begin{split}
        k_j(n) &= \left(1+\delta_{1,\alpha_1}\right)
        + \sum\limits_{k=1}^{n-1} \Bigg\{\delta_{0,\alpha_{k+1}}\cdot\sum\limits_{i=0}^{2^{k-1}-1} \left(\delta_{j-\sigma_{nk},\,\Xi_{ki}+1/2} + \delta_{j-\sigma_{nk},\,\Xi_{ki}+3/2}\right)\\
        &+ \delta_{2,\alpha_{k+1}}\cdot\sum\limits_{i=0}^{2^{k-1}-1} \left(\delta_{3^k +1 -j+\sigma_{nk},\,\Xi_{ki}+1/2} + \delta_{3^k +1 -j+\sigma_{nk},\,\Xi_{ki}+3/2}\right)
        +\delta_{1,\alpha_{k+1}}\cdot 2^k \left(\delta_{j-\sigma_{nk},\,1} + \delta_{j-\sigma_{nk},\,3^k}\right) \Bigg\}
    \end{split}
\end{equation}%
\end{widetext}%
where $\sigma_{nk}$ and $\Xi_{ki}$ are defined as:
\begin{align}
        \sigma_{nk} &= \sum\limits_{l=k}^n\alpha_{l+1} 3^l\,, \qquad
        \Xi_{ki} = \frac{3^k}{2} +\frac{3i}{2} +\sum\limits_{m=1}^i \frac{3^{a(m)}}{2}\,.
\end{align}
The $\alpha_i$'s, where $i\in\{1,...,n\}$ are the digits in the base 3 representation of the node index $j$. For example, for $j=17$ in generation $n=4$, the number $j$ can be written as $17 = 2\cdot 3^0 + 2\cdot 3^1 + 1\cdot 3^2 + 0\cdot 3^3$, hence, $\vec{\alpha} = (\alpha_1, \alpha_2,\alpha_3,\alpha_4) = (2,2,1,0)$.  Furthermore, the function $a(m)$ is the so-called \textit{ruler function} (OEIS A001511) \cite{allouche_shallit_2003}, where the $n$th term in this integer sequence is the highest power of 2 that divides $2 n$. It is clear from equation \eqref{eq:degree} that not every possible integer value of $k$ occurs in every generation. In particular, the number of distinct values for the degree grows as $2n+1$ for $n>1$. In a generation $n>1$, $k$ can only take values in the set $K^{(n)}$, where $K^{(n)} = \left\{1,2,\Lambda^{(n)}_{\text{odd}},\Lambda^{(n)}_{\text{even}}\right\}$, and where the sets of odd and even degrees are defined, respectively, as $\Lambda^{(n)}_{\text{odd}} = \left\{1+2^{\frac{j-1}{2}}\right\}_{j\, \text{odd}}$ and $\Lambda^{(n)}_{\text{even}} = \left\{2+2^{\frac{j}{2}-1}\right\}_{j\, \text{even}}$, with $j = 3, 4, ..., 2n-1$. For the first generation, the set of possible degrees is $ K^{(1)} = \left\{1,2\right\}$. 

After some algebra one can deduce the $n$th generation degree distribution $P^{(n)}(k)$ from the node degree as described by equation \eqref{eq:degree}, i.e.,
\begin{widetext}
\begin{align}
\label{eq:degree_distribution}
P^{(n)}(1)&=\frac{2}{3^n}\,, && P^{(n)}(2)=\frac{1}{3}\,, && P^{(n)}(3)=\left(\frac{2}{3}\right)^{n-1} - \frac{2}{3^n}\,,\\ P^{(n)}(4)&=\frac{5}{9} - \left(\frac{2}{3}\right)^{n-1} + \frac{1}{3^n}\,, && P^{(n)}\left(\Lambda^{(n)}_{\text{odd}}\right)=2^{\frac{1-j}{2}}\left(\frac{2}{3}\right)^n\,, &&
P^{(n)}\left(\Lambda^{(n)}_{\text{even}}\right) = 2\left(3^{-\frac{j}{2}} - 2^{-\frac{j}{2}}\left(\frac{2}{3}\right)^n\right)\,.\nonumber
\end{align}
\end{widetext}

In Fig. \ref{fig:degree_distribution}, the numerically determined degree distribution is shown for $n=10$ (red crosses), together with the theoretically predicted even (odd) degree distributions $P_e$ ($P_o$), for values of $n = 10$ (black dots) and $n=20$ (blue squares). For visualisation purposes, lines are shown for the even (full) and odd (dashed) distributions. 

Note that the degree distribution behaves as a power law, i.e., $P_{e} \sim k^{-\gamma_e}$, and $P_{o} \sim k^{-\gamma_o}$, with exponents $\gamma_e$ and $\gamma_o$. The exact values of $\gamma_e$ and $\gamma_o$ can be calculated in the large $n$ limit to be $\gamma_e = \log{3}/\log{2}\approx1.585$ and $\gamma_o = 1$. Both values were confirmed numerically, as indicated in the inset in Fig. \ref{fig:degree_distribution}. The value of $\gamma_e = \log{(3)}/\log{(2)}$ has been found previously for the deterministic network introduced by Barab\'asi, Ravasz, and Vicsek \cite{BARABASI2001559,Kazumoto2005}. Eliminating $j$ from $k = 1+2^\frac{j-1}{2}$ and $P^{(n)}\left(k\in\Lambda^{(n)}_{\text{odd}}\right)$, one finds that for large $k$, the odd degree distribution is $P^{(n)}\left(k\in\Lambda^{(n)}_{\text{odd}}\right) \sim (2/3)^n\, k^{-1}$. It follows that this vanishes for large $n$. A similar calculation reveals that no such behaviour is present for the even degrees, which will consequently dominate $P(k)$ for large $n$.

\begin{figure}[htp]
    \centering
    \includegraphics[width=0.95\linewidth]{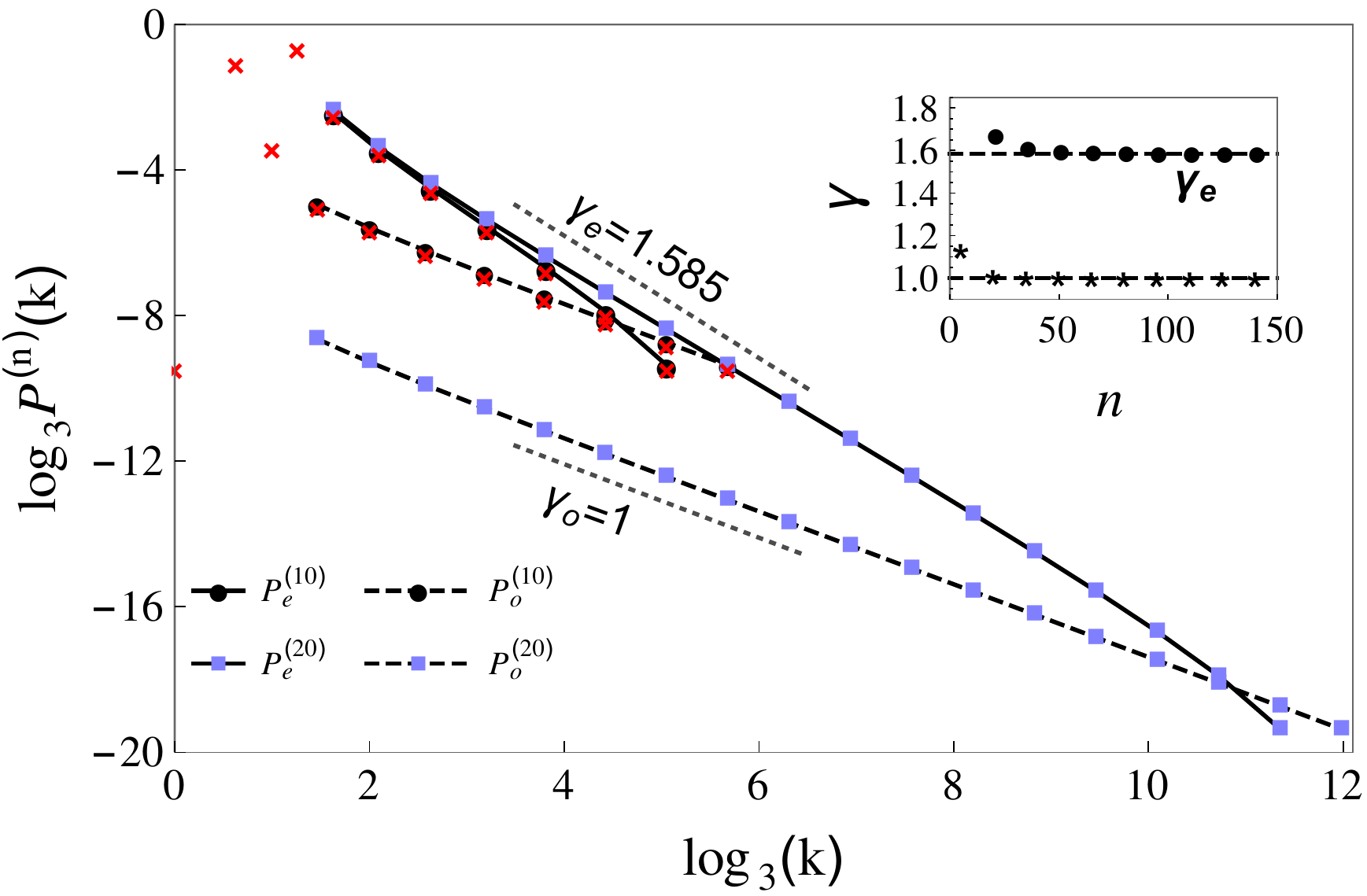}
    \caption{The numerical degree distribution (red crosses) for the $n=10$ HVG. $P_e(k)$ and $P_o(k)$ \eqref{eq:degree_distribution} are shown (respectively black full and dashed lines) for $n=10$ (black dots) and $n=20$ (blue squares). \textbf{Inset:} exponents $\gamma_e$ (dots) and $\gamma_o$ (stars) as a function of $n$.}
    \label{fig:degree_distribution}
\end{figure}

Equipped with the exact expression for the degree distribution \eqref{eq:degree_distribution}, one can calculate a number of properties. The moments of the distribution can be found as 
\begin{equation}
    \label{eq:moments}
    \langle k^m(n)\rangle = \sum_k k^m P^{(n)}(k)\,.
\end{equation}
Note that this expression only converges for $m<\gamma-1$, so only the mean degree is finite for the HVG associated with the HDDM. The mean degree $\langle k\rangle$ (where we suppressed $n$) can be determined exactly from equation \eqref{eq:moments} for $m=1$ as $\langle k\rangle = 4 \left(1-\left(2/3\right)^n\right)$. This expression converges to $\langle k^{\infty}\rangle =  4$, indicating that the network becomes sparse for large $n$, a property that could also be determined from the increasing sparsity of the adjacency matrix of the HVG. The mean degree can be used as a measure of the irregularity of the of the underlying data, since for an infinite periodic series
of period T, it is given by
\begin{equation}
    \label{eq:mean_degree_infinite}
    \langle k\rangle = 4 \left(1-\frac{1}{2T}\right)\,,
\end{equation}
which becomes $\langle k\rangle = 2$ for a constant series and $\langle k\rangle =4$ for a completely aperiodic series where $T\rightarrow\infty$ \cite{Nunez12}. The surface height profile hence becomes completely irregular in the long time limit; any periodicity in the model is washed out.  

Through a similar calculation we find that the second moment $\langle k^2\rangle$ diverges for large generations $n$ as $\langle k^2\rangle \sim 2(4/3)^n$, as expected for scale-free networks. Since scale-free networks arise naturally from fractal series in the context of VAs, one needs another measure to determine whether the underlying series is, e.g. a self-affine or self-similar fractal \cite{Lacasa4972}. One such measure is the network diameter $D$, which we now proceed to calculate exactly.

\subsection{Network diameter and clustering coefficient}
The diameter $D$ of a graph is defined as $D = \max_{ij}{l_{ij}}$, where $l_{ij}$ is the shortest path distance between nodes $i$ and $j$. For the HVG, the diameter $D(n)$ in generation $n$ is 
\begin{equation}
    \label{eq:diameter}
    D(n) = 2 (2^n -1)\,.
\end{equation}
This can easily be shown by considering that in the fractal landscape the outermost nodes $i=1$ and $j=3^n$ are the two most distant nodes in the corresponding HVG. The total path distance between these nodes can hence be split into two identical distances by virtue of the reflection symmetry about the central node. Therefore, we only consider the left side of the landscape for the calculation of the diameter. Assume now that we know the shortest path distance between the outermost left node and the central node in generation $n-1$. Then, again by virtue of the symmetry, we only need to consider the left side of the different copies of the $n-1$ landscape, because the leftmost node of such a copy is already visible from the central node of the copy to the left. 

This is shown graphically with blue lines in Fig. \ref{fig:diameter_HVG}. The resulting diameter of the left side of the landscape is twice the left diameter from the previous generation plus one from the connection between the two copies. From this construction, it is clear that $D(1) = 2$, $D(2) = 6$, $D(3) = 14$, ... Hence, the recurrence for the total diameter is $D(n+1) = 2D(n) +2$. Solving for $D(n)$ results in equation \eqref{eq:diameter}. It is clear that for large $n$, the diameter increases as a power law of the number of nodes $N(n) = 3^n$, i.e., $D(n)\sim N(n)^\epsilon$, making the HVG a self-similar, scale invariant fractal network \cite{Lacasa4972}. The exponent $\epsilon$ can easily be found to be the inverse of the degree exponent for the even degree nodes $\epsilon = \gamma_e^{-1}$. Such a power-law growth of the diameter is a signature of hub repulsion in fractal networks \cite{Lacasa4972,Song2006}, where the hubs of the network are the nodes with the highest connectivity. In our deterministic model, these hubs are located in the locally monotonic regions of the surface and are effectively repulsed from one another by the intermittent nodes corresponding to the highest data. Conversely, for stochastic self-affine series such as Brownian motion $B(t)$, where $B(t) = \lambda^{1/2} B(t/\lambda)$, the associated VG evidences the small-world property. Hence, the network diameter or associated quantities such as the average path length can be used to differentiate between different kinds of fractality in the underlying height profile.

By inspecting the clustering coefficient $C(k)$ for the nodes with degree $k$, it becomes clear that there is a one-to-one correspondence between the clustering coefficient of a vertex and its degree: $C(k)$ is proportional to $k^{-1}$ and stays stationary, i.e., nodes with degree $k$ always have the same clustering coefficient $C(k)$, independent of the generation. A simple geometrical argument reveals that the clustering coefficient is the ratio of the number of triangles that a node $i$ with degree $k$ is part of, normalised by the maximal number of triangles, i.e., ${k\choose2}$. So for e.g. $k=2$, the number of triangles is one (except for the central node of the network, which is not a part of any triangle). Hence, $C(k =2) = 1$. For $k=3$, one of the connections will be one of the neighbour data, while the other two form one triangle. Hence, $C(k=3) = 1/3$. In Table \ref{tbl:clustering}, the clustering coefficients are listed as a function of the node degree $k$ and generation $n$.
\begin{table}[h]
\centering
\begin{tabular}{|c|c|c|c|c|c|}
\hline
$k$ & $C^{(1)}(k)$ & $C^{(2)}(k)$ & $C^{(3)}(k)$& $C^{(4)}(k)$ & $C^{(5)}(k)$\\
\hline
2 & 0 & 1 & 1 & 1 & 1\\
3 & - & $1/3$ & $1/3$ & $1/3$ & $1/3$\\
4 & - & - & $1/2$ & $1/2$ & $1/2$\\
5 & - & - & $3/10$ & $3/10$ & $3/10$\\
6 & - & - & - & $1/3$ & $1/3$\\
9 & - & - & - & $7/36$ & $7/36$\\
10 & - & - & - & - & $1/5$\\
17 & - & - & - & - & $15/136$\\
\hline
\end{tabular}
\caption{Clustering coefficient $C^{(n)}(k)$ for the first five generations of the visibility graph.}
\label{tbl:clustering}
\end{table}
We find that the nodes with even and odd degree have the following form, respectively
\begin{align}
    \label{eq:clustering_k}
        C\left(k\in\Lambda^{(n)}_{\text{even}}\right) &= \frac{2}{k}, && 
        C\left(k\in\Lambda^{(n)}_{\text{odd}}\right) = \frac{2 (k-2)}{k (k-1)},
\end{align}
which are independent of the generation $n$. By construction, we also have $C(k=1) =0$, as nodes with only a single connection (edge nodes) can never form triangles. Hence, for large $k$, $C(k) \sim 2k^{-1}$, as was previously found to indicate hierarchical clustering \cite{Ravasz2002,Dorogovtsev2002}. The scaling law indicates that a hierarchy of nodes with different degrees of modularity coexists in the network. Since the local clustering coefficient is based on a three-point relationship, i.e., it is a function of three elements in the adjacency matrix, it expresses a local convexity property of the hierarchical surface \cite{Donner2012}. As a consequence, local minima of the surface tend to possess the highest clustering coefficient, as can for example be seen for internal nodes where $h=0$, which always have $k=2$ and thus $C = 1$, since the nearest neighbours of these nodes must be mutually visible for the node to be a local minimum. The clustering properties of the visibility graph can consequently be used as a measure for the local convexity of the underlying surface.

The mean clustering coefficient $\langle C^{(n)}\rangle$ can be found by 
\begin{equation}
    \label{eq:mean_clustering}
    \begin{split}
        \langle C^{(n)}\rangle &= \sum_k C(k) P^{(n)}(k)\,,
    \end{split}
\end{equation}
For large $n$, the mean clustering coefficient saturates to a constant value $\langle C^\infty\rangle$. It can easily be shown that the contributions from the odd degrees vanish for large $n$, as well as the majority of factors from the even degrees. The terms that remain are
\begin{equation}
    \label{eq:mean_clustering_limit}
    \begin{split}
        \langle C^\infty\rangle &= \lim\limits_{n\rightarrow\infty}\langle C^{(n)}\rangle = \frac{11}{18} + 8\sum\limits_{l=3}^\infty\frac{3^{-l}}{4+2^l}
    \end{split}
\end{equation}
where in the sum over the even degrees we have made the change of variables $l=j/2$. It can be checked numerically that $\langle C^n\rangle$ converges to the value $\langle C^\infty\rangle \approx 0.642$. This is shown in Fig. \ref{fig:clustering} together with the relations \eqref{eq:clustering_k}.

\begin{figure}[htp]
    \centering
    \includegraphics[width =0.95\linewidth]{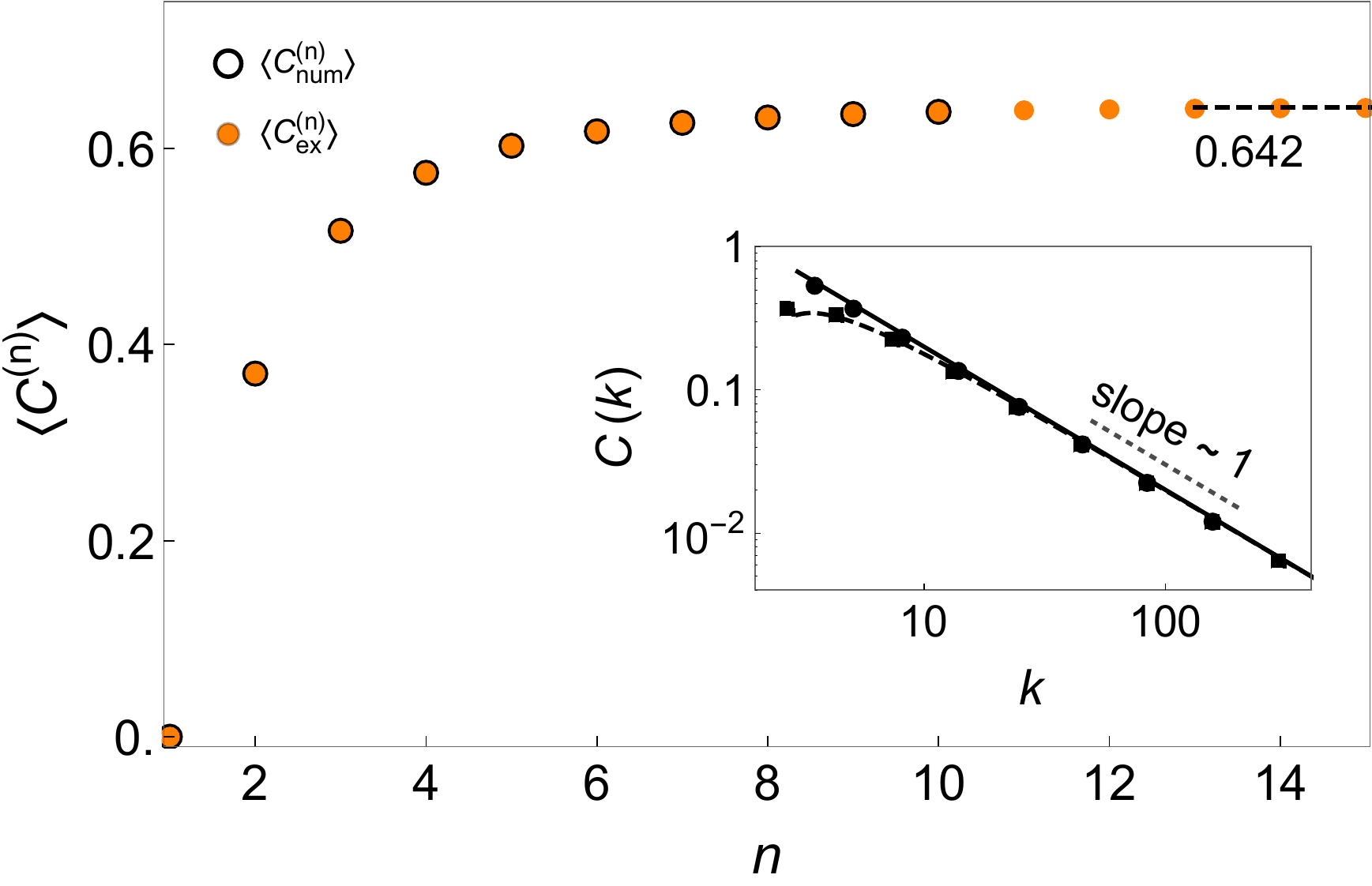}
    \caption{Numerical (black) and exact (orange) mean clustering coefficient $\langle C^{(n)}\rangle$ as a function of $n$. The black dashed line indicates $\langle C^\infty\rangle\approx0.64185$. \textbf{Inset:} $C(k)\sim2 k^{-1}$ for even (full line) and odd degrees (dashed line).}
    \label{fig:clustering}
\end{figure}

\subsection{Adjacency matrix and the eigenvalue spectrum}

The adjacency matrix $A^{(n)}$ is the matrix with values $a_{ij} = 1$ if the nodes $i$ and $j$ have mutual visibility, and value $a_{ij} = 0$ otherwise. The $A^{(n+1)}$ matrix is a block tridiagonal matrix with on the diagonal the adjacency matrix $A^{(n)}$ from the previous generation, repeated three times. It has the general form:
\begin{equation}
    \label{eq:adjacency_matrix_form}
    A^{(n+1)} = \begin{pmatrix}
        A^{(n)} & C^\intercal & 0\\
        C & A^{(n)} & B^\intercal\\
        0 & B & A^{(n)}
    \end{pmatrix}
\end{equation}
The matrix $B$ is related to $C$ by $B = J C^\intercal J$, where $J = antidiag(1,...,1)$ is the $3^n\times3^n$ exchange matrix.

For the adjacency matrix spectrum, we follow the representation used in \cite{Flanagan_2019} and rescale the eigenvalue index for every generation $n$ such that the smallest eigenvalue corresponds to 0 and the largest corresponds to 1. This way, we can show that the spectrum converges to a fixed shape with a hierarchical structure. In Fig. \ref{fig:evplot}(a), we show the spectrum in the rescaled coordinates for the first seven generations, where the point size is inversely related to the generation number. Note that the shape of the spectrum appears to be fractal, which we can surmise to be a consequence of Cauchy's interlacing theorem, as $A^{(n)}$ is always a proper subgraph of $A^{(n+1)}$. However, the exact shape of the curve eludes any analytic description so far and will be left for future work. The largest eigenvalue $\lambda_{max}$ of the adjacency matrix can be numerically determined to depend on the system size $N=3^n$ as $\lambda_{max} \sim N^\frac{1}{4} \sim 3^\frac{n}{4}$. In Fig. \ref{fig:evplot}(b), the eigenvalue distribution $\rho(\lambda)$ is drawn for $n=8$, showing that the distribution possesses power-law tails where $\rho(\lambda)\sim |\lambda|^{4.04}$. The scaling of the largest eigenvalue and the tail of the eigenvalue spectrum is very similar to the results previously found for the Barab\'asi-Albert scale-free network \cite{Goh2001}.

\begin{figure}[htp]
    \centering
    \includegraphics[width=\linewidth]{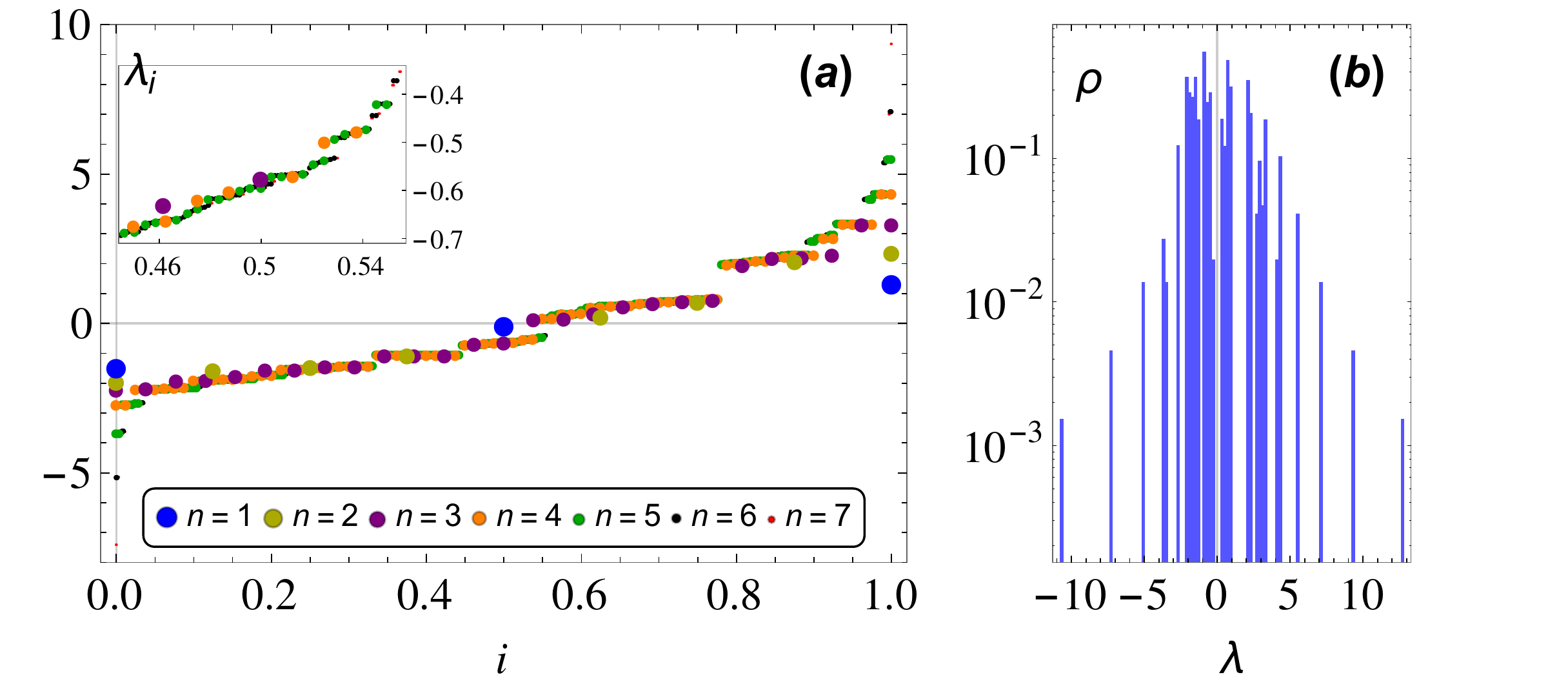}
    \caption{\textbf{(a)} Eigenvalue spectrum $\{\lambda_i\}$ as a function of the rescaled coordinate $i$ for generations $n = 1,...,7$. The spectrum converges to a fixed shape for large $n$. \textbf{Inset:} zoomed image around the central part of the spectrum. The fine structure is clearly visible. \textbf{(b)} Eigenvalue distribution $\rho(\lambda)$ for $n=8$ in semilog scale.}
    \label{fig:evplot}
\end{figure}

The largest eigenvalue of a visibility graph's adjacency matrix has been proposed as a measure to characterise the complexity of the associated sequence by means of the so-called Graph Index Complexity (GIC) \cite{KIM20082637}, which is a rescaled version of $\lambda_{max}$. It has been used for example as a measure to discriminate between randomness and chaos in the underlying data of the graph \cite{fioriti,Flanagan_2019}. In this work, we will not elaborate on this subject further but leave a detailed analysis of the spectral properties for future investigation.

We now move on from the HDDM and introduce an element of randomness into the process. We repeat our study of the graph-theoretical quantities.

\section{Hierarchical random deposition}\label{sec:HRD}
Let us generalise the model introduced in the previous section by introducing an element of randomness. In a generation $n$, a number $\lambda^n$ of disjoint sites are available. These sites are either located on the substrate, if no deposition event has occurred on this site prior to generation $n$, or on previously deposited particles. A square particle with linear size $s = \lambda^{-n}$ is deposited on every site with a probability $P\in[0,1]$. Additionally, we allow for the erosion of a block of similar size that is located on the interface with probability $Q\in[0,1]$, where the total probability of either deposition or erosion is bounded, i.e., $0<P+Q<1$. Erosion events thus remove material that has either been deposited in prior generations, or is part of the substrate. Lateral attachment of incoming particles, i.e., ballistic deposition, is not allowed and consequently overhangs cannot be formed. This process has been named the hierarchical random deposition model (HRDM). The parameters $P$ and $Q$ allow us to tune the amount of noise in the growth process. For $Q=0$ and $P\rightarrow0$ (i.e., the dilute limit), classical sequential deposition is recovered \cite{Krug1991}. 

The HVA can now be applied in exactly the same manner as was outlined in section \ref{sec:HDD}. In Fig. \ref{fig:HVG_random}, we illustrate the landscape formed by random deposition and the horizontal visibility lines between the data, as well as the associated complex network. Due to the randomness introduced in the HRDM, analytical progress is challenging so we will proceed numerically. However, we show analytically that the fractal nature of the set of admissible height values remains fractal (for large $n$) even when randomness is introduced into the model.

\subsection{The degree distribution}
Based on the seminal work on HVGs performed in Ref. \cite{Luque2009}, one would expect that the degree distribution becomes exponential due to the random nature of the deposition. There are, however, some caveats to such a claim. In the following, we show that the degree distribution is scale-free with exponent $\gamma =3$.
\begin{figure}
    \centering
    \includegraphics[width=\linewidth]{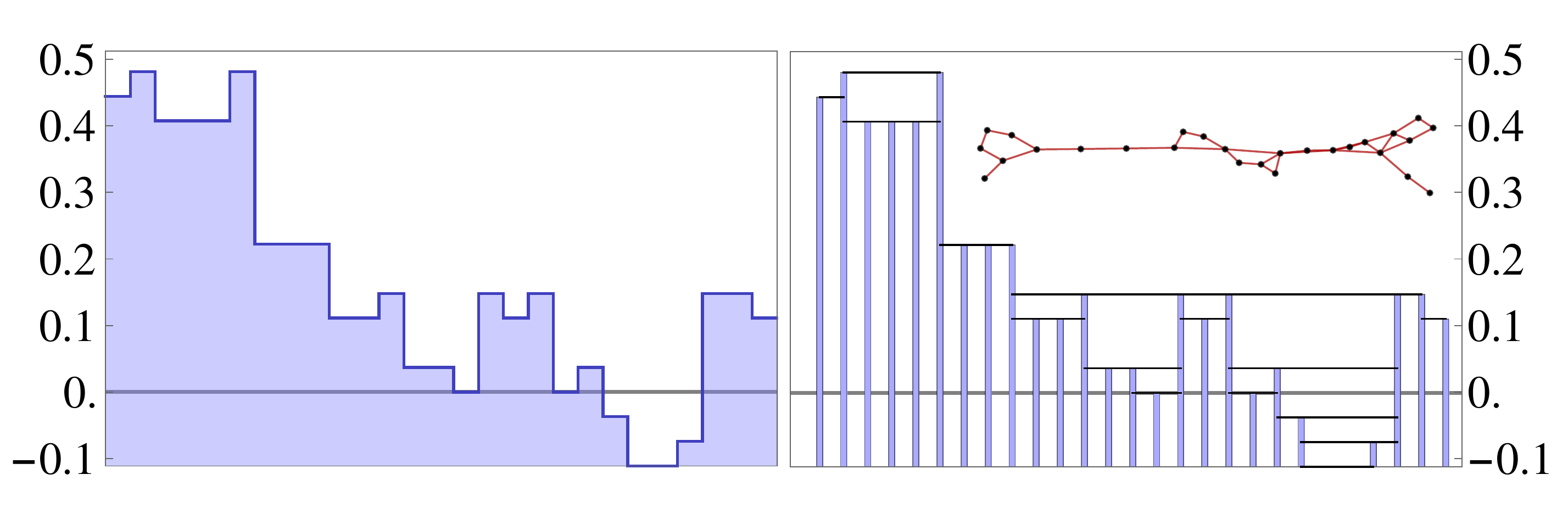}
    \caption{The landscape formed by hierarchical random deposition with $P = 0.5$, $Q =0.1$, $n=3$ and $\lambda=3$ (left) and the associated visibility lines (right). The inset in the right figure is the complex network formed by applying the horizontal visibility criterion \eqref{eq:visibility}.}
    \label{fig:HVG_random}
\end{figure}
In their initial work \cite{Luque2009}, the authors proved that the degree distribution associated with a bi-infinite (time) series generated from an uncorrelated random variable $X$ with an absolutely continuous probability density function $f(x)$ with $x\in[0,1]$ has the following exponential form:
\begin{equation}
    \label{eq:exponential}
    P(k) = \frac{1}{3} \left(\frac{2}{3}\right)^{k-2} = \frac{1}{3} \me^{-\mu_c (k-2)}\,,
\end{equation}
where $\mu_c = \ln{(3/2)}$. Note that for finite $n$, the set of admissible heights in the hierarchical deposition model is also finite, and the height distribution function is hence discrete. Therefore, one cannot \textit{a priori} expect an exponential degree distribution for $n<\infty$.

However, for uncorrelated random variables taking \textit{discrete} integer values up to a number $m$, an exact formula for $P(k)$ can be found, which does not possess a generic algebraic closed form \cite{Lacasa2016}. For $m\rightarrow\infty$, the degree distribution reduces to the exponential form that was found for the continuous case. The premise of the authors' results was based on the equispacedness of the support of the distribution, i.e., for the field $\mathbb{F} = (h_1, h_2,\dots,h_n) \subset \mathbb{N}$, it holds that $h_i = h_{i-1} + c$, $c\in\mathbb{N}^+$, for $m<\infty$, since the HVG is invariant under monotonic transformations in the series.
For asymptotically unbounded support, i.e., the continuous case, one can interpret an unequispaced subset of the natural numbers simply as the set of natural numbers (equispaced) where some elements have probability zero. Since in the continuous case the marginal distribution does not matter, one recovers the exponential distribution \eqref{eq:exponential} even when the support seems to be unequispaced. For finite $m$, however, it is not clear how the unequispacedness affects the shape of the degree distribution. 

In the HRDM, however, the height of a site can take discrete values that, when scaled with $\lambda^n$, are elements of a growing subset of $\mathbb{Z}$, where the increments between elements are unequispaced. In appendix \ref{app:B}, we calculate that in the long-time limit, $n\rightarrow\infty$, the cumulative height distribution $\mathcal{F}(h) = \lim_{n\rightarrow\infty}\mathcal{P}^{(n)}(H\leq h)$ is given by
\begin{widetext}
\begin{equation}
    \label{eq:cumulative_full2}
    \mathcal{F}(h) = \frac{1}{2\pi i}\int_\mathbb{R} \mathrm{d}t\,\frac{\me^{it/(\lambda-1)} - \me^{-ith}}{t}\prod_{j=1}^\infty\left[P \me^{it\lambda^{-j}} + Q
    \me^{-it\lambda^{-j}} + (1-P-Q)\right]\,.
\end{equation}
\end{widetext}
Note that although $\mathcal{F}$ is continuous, the height distribution function $\mathcal{P}(H = h) = \lim_{n\rightarrow\infty}\mathcal{P}^{(n)}(H = h)$ is a generalisation of the Cantor distribution and is therefore not absolutely continuous with respect to Lebesgue measure, making it a singular distribution. The set of final accessible heights is a fractal due to the retrodictive nature of the deposition process, i.e., every
final height can be reached by a unique path \cite{delatorre2000}.

These properties constitute the major differences between our results and the results in Ref. \cite{Lacasa2016}. Since the set of admissible (scaled) height values is a finite and non-uniformly spaced subset of $\mathbb{Z}$ for $n<\infty$, and their distribution depends on $P$, $Q$ and $\lambda$, the degree distribution of the HVG associated with the hierarchical random deposition model cannot generally be mapped one-to-one to the one found in Ref. \cite{Lacasa2016}. Asymptotically, the height distribution is singular, so the degree distribution cannot be expected to be exponential. Since we have established that the set of admissible height values is a fractal for $n\rightarrow\infty$ and $P+Q\neq0$, the degree distribution is expected to converge to a power law $P(k)\sim k^{-\gamma}$.

For $n\rightarrow\infty$, we can find the degree exponent of this power law distribution exactly. In Ref. \cite{Lacasa_2009}, the authors showed empirically that the relation between the exponent $\gamma$ of the visibility graph degree distribution and the Hurst exponent $H$ (alternatively, the local roughness exponent) of the associated data series for fractional Brownian motion (fBm) and fractional Gaussian noise (fGn) is given by respectively
\begin{align}
    \gamma &= 3 - 2H\,, && \text{for fBm}\\
    \gamma &= 5 - 2H\,, && \text{for fGn}\,.
\end{align}
For self-affine processes, the Hurst exponent is related to the fractal dimension $D_f$ through $D_f = 2-H$, leading to the relations between $\gamma$ and $D_f$, i.e.,
\begin{align}
    \gamma &= 2D_f -1\,, && \text{for fBm}\\
    \gamma &= 2D_f +1\,, && \text{for fGn}\,.
\end{align}
For the HRDM, the fractal dimension is equal to the Euclidean dimension \cite{INDEKEU1998294}, so the degree exponent can be either $\gamma = 1$ (fBm) or $\gamma = 3$ (fGn), since $D_f =1$. However, the mean degree for HVGs is in fact bounded from above by $\langle k\rangle \leq 4$ \cite{Luque2009,GUTIN20112421}. With $\gamma = 1$, the mean degree would become unbounded. Hence, we can conclude that in the long time limit, the degree exponent converges to $\gamma=3$. Alternatively, a simple detrended fluctuation analysis (DFA) reveals that indeed $H=1$ for both the deterministic and random HDM.

Note that since the degree exponent is directly related to the topological dimension for the HRDM, it is independent of $P$, $Q$ and $\lambda$. The degree exponent thus presents an alternative approach for the calculation of the local roughness exponent and for the characterisation of the underlying process as either fBm or fGn . Note that the concomitant global roughness exponent cannot be determined in this manner, since surfaces can exhibit anomalous scaling \cite{DasSarma1996,Lopez1997} where the global and local roughness exponents are unequal. Moreover, it should be noted that for the HRDM the roughness is a power law of only the length already in the first generation. Consequently, there is no initial growth that can be found as a power law of time \cite{BERX2021125998}. Hence, the usual scaling exponents $\beta$ or $z$ are nonexistent in this model. A detailed treatment of this problem can be found in Ref. \cite{PhdGiuraniuc}.

To numerically find the degree exponent $\gamma$ of the scale-free networks generated by the HVA, we solve the following transcendental equation for the maximum likelihood estimator (MLE) $\hat{\gamma}$ \cite{Clauset2009}:
\begin{equation}
    \label{eq:MLE}
    \frac{\zeta'(\hat{\gamma},k_\text{min})}{\zeta(\hat{\gamma},k_\text{min})} = -\frac{1}{m}\sum\limits_{i=1}^m \ln{k_i}\,,
\end{equation}
where $\zeta$ is the Hurwitz zeta function \cite{nist}, $m$ is the total number of observed values of the degrees $k$, and $k_\text{min}$ is the cutoff degree above which the distribution is a discrete power law. The prime denotes differentiation with respect to the first argument. To remove boundary effects for small system size, we impose periodic boundary conditions, but only on the two outer nodes, which now neighbour each other. This fixes the minimal possible degree to be $k=2$. We combine our analysis with a Kolmogorov-Smirnov test \cite{Clauset2009} on the cumulative degree distribution to find the optimal $\hat{\gamma}$ and $k_\text{min}$ that fit the simulation data. In appendix \ref{app:C}, we list the exponents for increasing $n$, together with the numerical errors.

In Fig. \ref{fig:distribution_evolution}(a), the degree distribution is shown for increasing generations, for $P=Q=0.25$ and $\lambda=3$, averaged over 5000 realisations, indicating that the network is scale-free. In the inset, the exponent $\gamma$ is shown as a function of the generation $n$. An exponential fit yields the expression $\gamma(n) = 3.00213 + 12.2109\, \me^{-0.34159 n}$, which supports the theoretical assertion that the exponent equals $\gamma =3$ for $n\rightarrow\infty$. The dependence of $\gamma$ on $P$ and $Q$ is shown in Fig. \ref{fig:HRDM_quantities}(a). While the degree exponent converges to $\gamma=3$, the speed of this convergence depends on the deposition and erosion probabilities.

\begin{figure}[htp]
    \centering
    \begin{subfigure}{0.485\linewidth}
    \includegraphics[width=\linewidth]{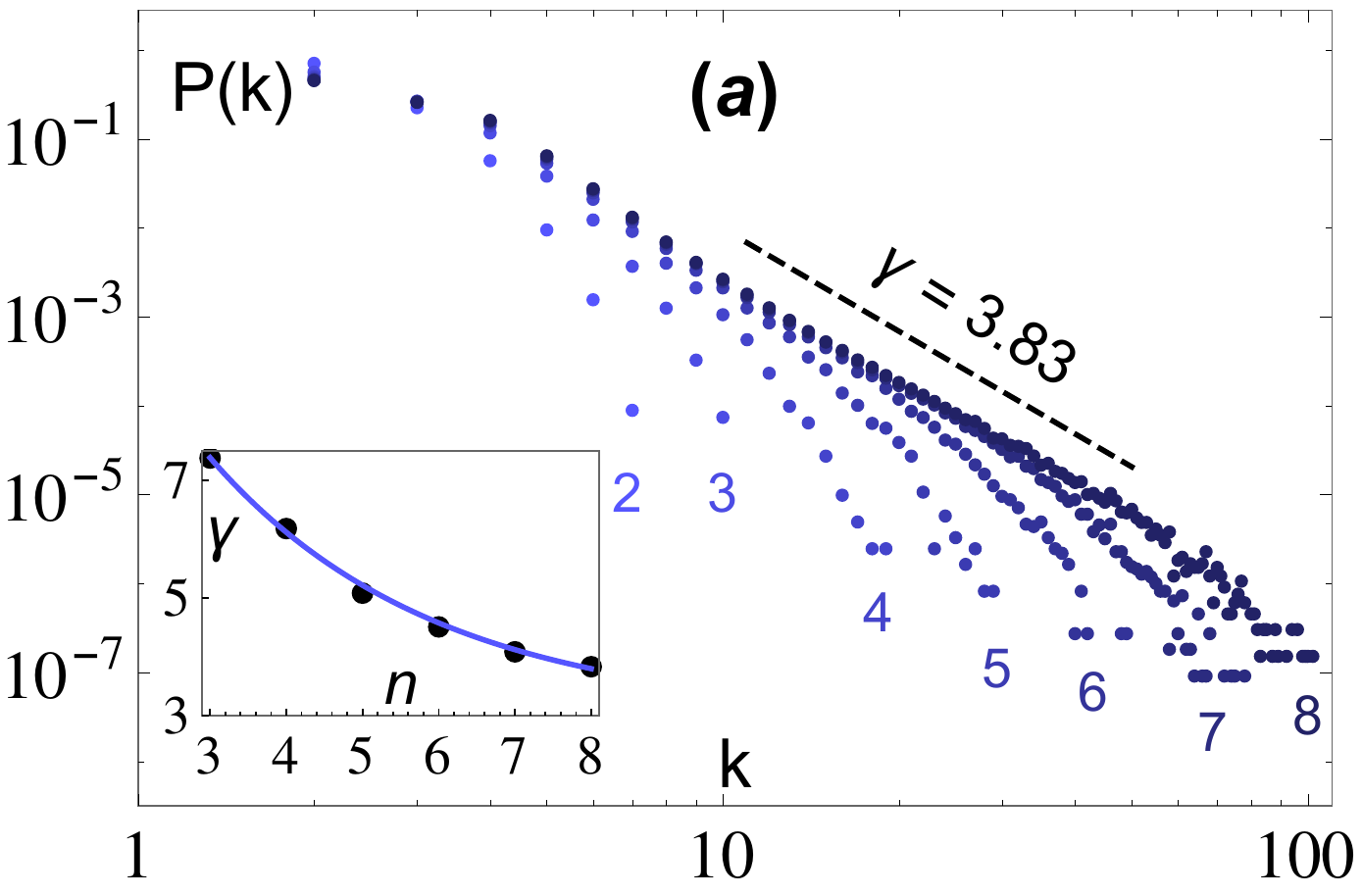}
    \end{subfigure}
    \begin{subfigure}{0.485\linewidth}
        \includegraphics[width=\linewidth]{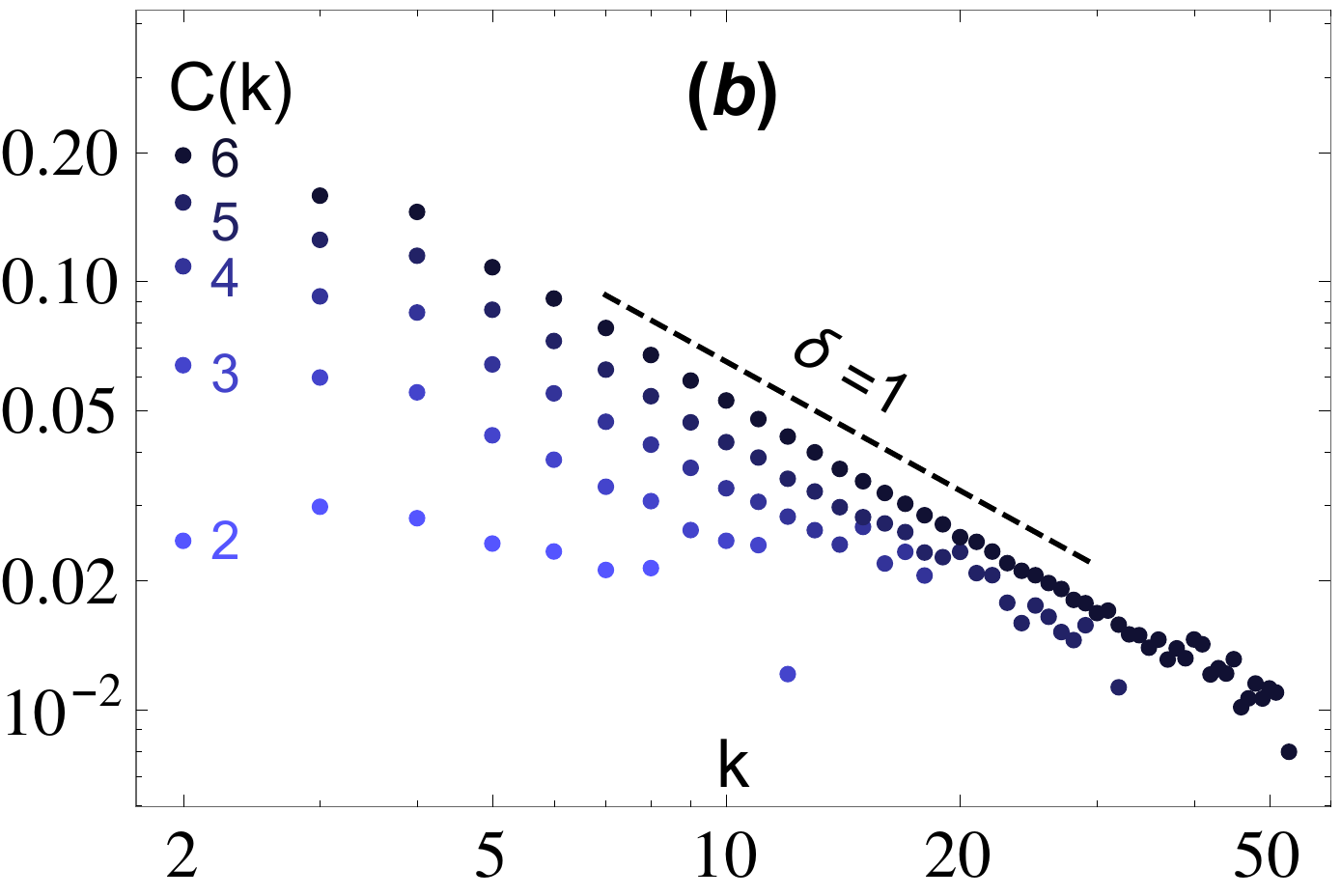}
    \end{subfigure}
    \caption{\textbf{(a)} The HVG degree distribution for increasing $n$, with $P = Q = 0.25$, and $\lambda=3$. \textbf{Inset:} Evolution of $\gamma(n)$. Symbols indicate numerical results, the full line is an exponential fit. \textbf{(b)} The degree dependence of the clustering coefficient $C(k) \sim k^{-\delta}$ for increasing $n$ with $P=Q=0.3$, shifted vertically for visibility purposes. Results are averaged over 5000 realisations.}
    \label{fig:distribution_evolution}
\end{figure}

The mean degree $\langle k\rangle$, degree variance $\sigma^2_k$, and mean clustering coefficient $\langle C\rangle$ are shown in Fig. \ref{fig:HRDM_quantities}(b)-(d). It can be seen that for $P=0$ and $P=1$, the mean degree is $\langle k\rangle = 2$ and the variance and mean clustering are $\sigma_k^2 = \langle C\rangle =0$, since the surface height profile is flat and the associated HVG becomes regular. For these points, the degree distribution is trivially $P(k) = \delta_{k,2}$ and hence the degree exponent $\gamma$ diverges to infinity. For a scale-free network with exponent $\gamma$, the mean degree and degree variance are
\begin{align}
    \langle k\rangle &= \frac{\zeta(\gamma-1)-1}{\zeta(\gamma)-1}\label{eq:mean_degree_zeta}\\
    \sigma_k^2 &= \frac{\zeta(\gamma-2) \left[\zeta(\gamma)-1\right]-\zeta(\gamma-1)\left[\zeta(\gamma-1)-2\right]-\zeta(\gamma)}{\left(\zeta(\gamma)-1\right)^2}\,,
\end{align}
respectively. For $\gamma=3$, this evaluates to $\langle k\rangle \approx 3.2$ and $\sigma_k^2\rightarrow\infty$. For the clustering coefficient $C$ no such general analytical expression is known, since it depends on the specific wiring on the level of vertex neighbourhoods, and hence on the mechanism that formed the network. When we again consider the relationship between the local clustering coefficients $C_i$, $i =1, 2, \dots,\lambda^n$ and associated vertex degrees, we find that for one degree $k$, multiple values of $C_i$ are possible. Hence, to study $C(k)$, we average the $C_i$ that correspond to one value of $k$. One can find that $C$ exhibits a power-law relation, i.e., $C(k)\sim k^{-\delta}$, with $\delta>0$. For increasing $n$, the exponent becomes $\delta = 1$. This is shown in Fig. \ref{fig:distribution_evolution}(b). The hierarchical clustering that was present for the HDDM persists even when randomness is introduced.

It can be shown numerically (see appendix \ref{app:C}), by means of a Floyd-Warshall algorithm \cite{cormen2001introduction}, that the network diameter is again a power law of the number of vertices, i.e., $D(N)\sim N^\epsilon$. This indicates that the underlying height profile is self-similar and not self-affine, the same result that was found in Section \ref{sec:HDD} for the HDDM. The scaling exponent $\epsilon$ now depends on the system parameters but always takes a value $0<\epsilon\leq 1$. The maximal value of $\epsilon = 1$ is reached for flat profiles where the diameter is maximal, since it is the distance between the edge nodes of the system, i.e., $D(N) = N-1$.

From Fig \ref{fig:HRDM_quantities}, we find that for a quantity $X(P,Q)$, with $X\in\left\{\gamma,\langle k\rangle, \sigma^2_k, \langle C\rangle\right\}$, the following invariance properties hold,
\begin{align}
    X(1-a,a) = X(1-a,0) = X(0,a) \label{eq:invariance1}
\end{align}
with $0\leq a\leq1$. Since this invariance holds for $\gamma$ (and trivially for $\langle k\rangle$ and $\sigma_k^2$), we can conclude that the degree distribution is identical for parameters choices that obey the invariance, even for finite generation $n$.
\begin{figure}[htp]
    \centering
    \includegraphics[width=0.9\linewidth]{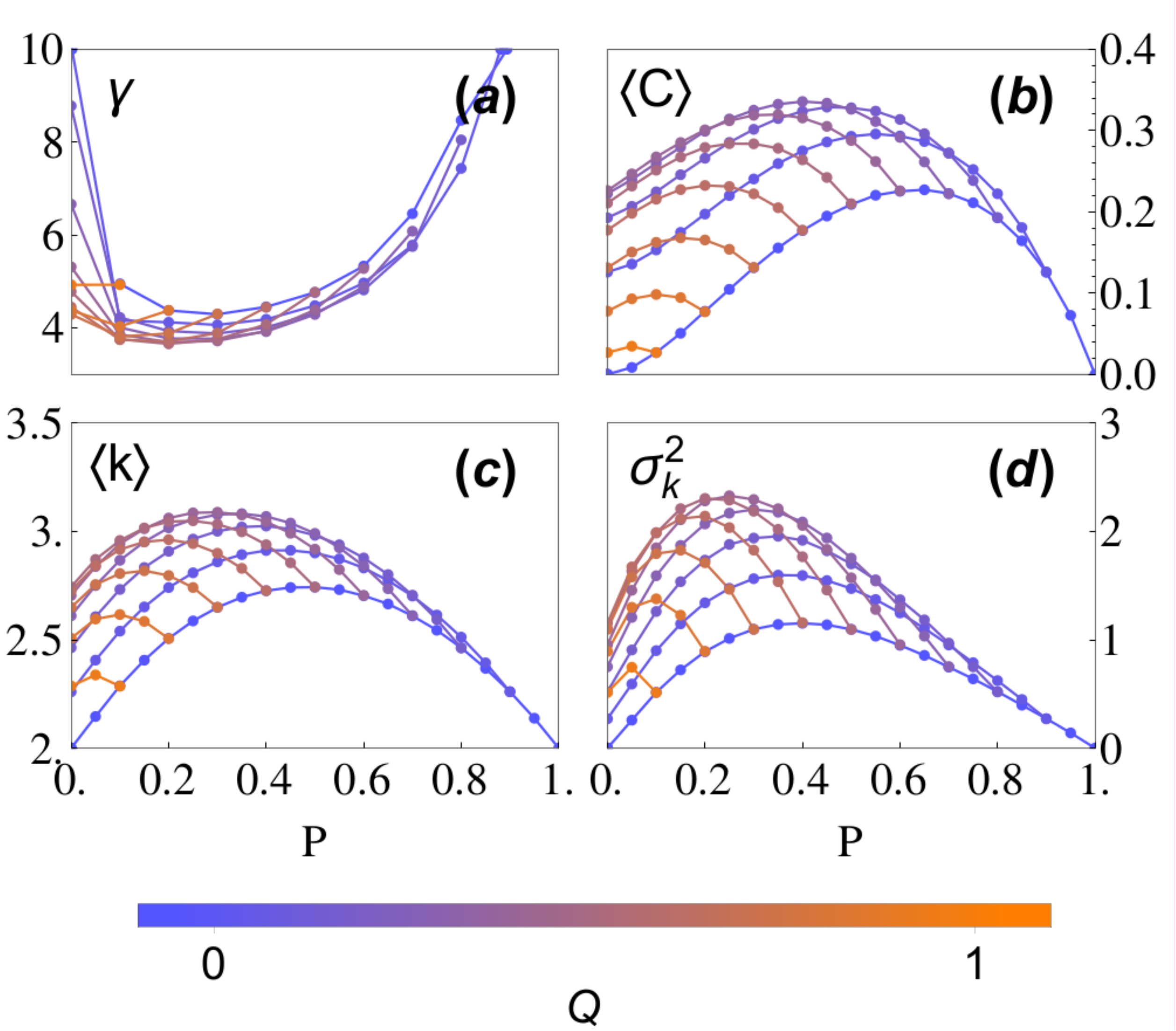}
    \caption{\textbf{(a)} The degree exponent $\gamma$ as a function of $P$ for increasing $Q$ for $n=8$. \textbf{(b)-(d)} Mean clustering coefficient $\langle C\rangle$, mean degree $\langle k\rangle$ and variance $\sigma^2_k$ as a function of $P$ for increasing $Q$, with $n=5$. Results are averaged over 5000 realisations.}
    \label{fig:HRDM_quantities}
\end{figure}

The different network quantities shown in Fig. \ref{fig:HRDM_quantities} can be used to study the structure of the underlying surface profile. Let us for example inspect the mean degree $\langle k\rangle$ as a function of $S=P+Q$, i.e., the total probability that a deposition or erosion event occurs at each site. Since $S$ is degenerate for different combinations of $P$ and $Q$, we average quantities with $S$ fixed. One can immediately see (appendix \ref{app:C}) that the highest value of $\langle k\rangle$ occurs for $S\approx0.6$. Similarly, the highest mean clustering coefficient, highest degree variance, and lowest degree exponent also occur for $S\approx0.6$. Intuitively, one can expect that the mean degree and clustering coefficient are maximal at the same value of $S$, since the higher average convexity of the surface, characterised by a high value of $\langle C\rangle$, implies that it is more irregular, which leads to a high value of $\langle k\rangle$. Moreover, since the mean degree of a power-law distribution only depends on the degree exponent $\gamma$, we can see that when the mean degree grows closer to the asymptotic value $\langle k\rangle \approx 3.2$ the degree exponent decreases to $\gamma\rightarrow3$, explaining the lower value of $\gamma$ for $S\approx 0.6$.

Before concluding, one remark is in order. While for $\lambda>2$ the deposition history is unique, this is not the case anymore for $\lambda=2$. Two different deposition histories can lead to the same landscape after two generations by means of partial levelling of vertical segments \cite{INDEKEU1998294}, where subsequent deposition and erosion of neighbouring blocks annihilates the vertical segment in between; the model thus becomes non-retrodictive. This degeneracy can possibly impact the properties of the associated HVG and its network topology. We defer to future research for a more detailed study.

\section{Conclusion and outlook}\label{sec:conclusions}
A mapping between the HDM for surface growth and complex networks was introduced, based on the HVA. We have shown that the resulting network is sparse and scale-free through the calculation of the degree distribution and exponents, and by calculating the first two moments of this distribution. We first performed these calculations for a deterministic version of the HDM and subsequently extended them to random deposition with the addition of a erosion probability. For the deterministic version, we then calculated the graph diameter and derived an exact expression for the clustering coefficient, which indicate that the network is fractal, self-similar, and exhibits hierarchical clustering. The spectrum of the adjacency matrix was determined numerically and it was shown that it converges to a yet undetermined fractal shape and its spectral density possesses power-law tails. For the random deposition model, we studied the resulting HVG numerically and showed analytically that $\gamma\rightarrow3$ asymptotically, indicating that the HRDM belongs to the fractional Gaussian noise stochastic processes. We showed that the hierarchical clustering and modularity found for the deterministic model is robust when introducing randomness into the system, and that it remains fractal through the scaling of the network diameter.

The connection between surface growth phenomena and complex networks is still largely unexplored. The results from this work indicate that the structure of the growing surface is encoded in the topology of the associated visibility network, and that it can be discovered through careful analysis of different network measures. Subsequent research could extend the above results to more complex situations, such as the HRDM with lateral attachment, resembling a hierarchical variation of the discrete model underlying the KPZ equation (i.e., ballistic deposition). Such extensions present an interesting problem. Since a point on the surface can possibly ``see'' both another point and an accompanying overhang, the associated HVG can be a multigraph.

More general, one can construct a classification of ``network universality classes'' for surface growth phenomena based on topological properties and exponents determined from visibility graphs. This can complement or refine the existing theory of universality classes for which the exponents are well-known. In particular, future research could extend the VA formalism to growth processes (or stochastic processes in general) that exhibit anomalous scaling \cite{mandelbrot2002gaussian,Chen2017}.

Another subject of future study can include the search for the other scaling exponents usually associated with dynamical scaling, i.e., the growth exponent $\beta$ and dynamic exponent $z$, and the corresponding relations between them. Since both exponents are concerned with the dynamical properties of the surface, it may be possible to find explicit time-dependent network quantities in the HVG and relate specific network growth exponents to $\beta$ and $z$. Mapping the dynamical scaling relations back to these network exponents can reveal hitherto unknown associations between them.

While the visibility algorithm approach to interface growth presents an interesting avenue of research, classical methods to calculate, e.g., growth exponents are still relatively easy to implement and do not usually necessitate the use of other methods. Hence, more future research is needed to identify the areas where VAs can really make a difference, either by facilitating calculations or by decreasing computational load.

\begin{acknowledgments}
The author is grateful for the spirited discussions with J. O. Indekeu and R. Tielemans, whose insights were very enlightening, and for the suggestions received from L. Lacasa. He also acknowledges the support of the G-Research grant.
\end{acknowledgments}

\appendix

\section{Derivation of the individual node degree for the HDDM}\label{app:A}
Consider the adjacency matrix of the visibility graph in generation $n$, i.e., $A^{(n)}$ with dimensions $3^n\times3^n$, see e.g., Fig. \ref{fig:adjacency_row}. There are $3^n$ nodes with degree $k_j$, where $j=1,2,...,3^n$ denotes the node index. Let us for convenience shift this index $j$ by one, $j' = j-1$ and expand the result into base 3, writing the digits to a vector $\vec{\alpha}_j$. For example, for the eighteenth node, $j=18$, hence $j'=17$ in generation $n=4$, and the number $j'$ can be written as $17 = 2\cdot 3^0 + 2\cdot 3^1 + 1\cdot 3^2 + 0\cdot 3^3$, hence, $\vec{\alpha}_{18} = (\alpha_1, \alpha_2,\alpha_3,\alpha_4) = (2,2,1,0)$, where every $\alpha_i\in\{0,1,2\}$ for $i=1,2,...,n$. Thus, every degree $k_j$ acquires a unique vector $\vec{\alpha}_j$, which is associated with its position in the adjacency matrix. We now consider the lowest-order adjacency matrix $A^{(1)}$, which is given by
\begin{equation}
\label{eq:adjacency_1}
    A^{(1)} = \begin{pmatrix}
        0 & 1 & 0\\
        1 & 0 & 1\\
        0 & 1 & 0
    \end{pmatrix}\,.
\end{equation}
For every subsequent generation $n$, the adjacency matrix $A^{(1)}$ represents a ``unit cell'', which is repeated along the diagonal of every $A^{(n)}$. From a physical standpoint, this can be attributed to the mutual visibility of neighbors, where every node can always see its direct neighbors. 

\begin{figure}[htp]
    \centering
    \includegraphics[width=\linewidth]{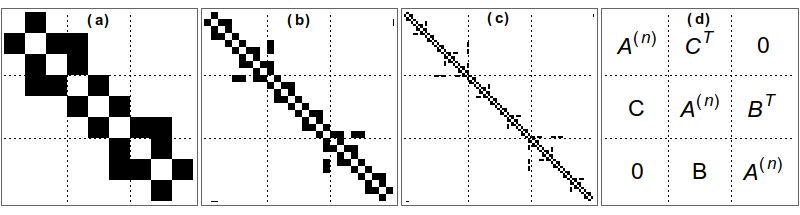}
    \caption{\textbf{(a)-(c)} Adjacency matrices for generations $n=2,3,4$, respectively. Black squares indicate matrix elements $a_{ij} = 1$, while white regions indicate that $a_{ij} = 0$. \textbf{(d)} A general depiction of the hierarchical structure of the adjacency matrix $A^{(n+1)}$.}
    \label{fig:adjacency_row}
\end{figure}

The first element in the vector $\vec{\alpha}$ indicates the column of the node index within this unit cell, i.e., $\alpha_1=0,1,2$ indicates the left, middle and right columns, respectively. The contribution of the position in this unit cell to the total node degree is the sum of elements in the column $\alpha_1$, and can be written consisely as $1+\delta_{\alpha_1,1}$, where $\delta_{ij}$ is the Kronecker delta, i.e., left and right columns contribute $1$ to the degree, while the middle column contributes $2$, which can easily be deduced from equation \eqref{eq:adjacency_1}.

We now look at higher ``levels'' $k$ in the adjacency matrix through the subsequent values in the $\vec{\alpha}$ vector; when $\alpha_k=0,1,2$, the index is located in the left, middle or right column of the $k$th level of the adjacency matrix. Let us look at the following example for $n=2$:
\begin{equation}
\begin{split}
    \label{eq:adjacency_2}
    A^{(2)} &= \left(\begin{array}{ccc}
        A^{(1)} & (C^{(1)})^t & 0\\
        C^{(1)} & A^{(1)} & J^{(1)} C^{(1)} J^{(1)}\\
        0 & J^{(1)} (C^{(1)})^t J^{(1)} & A^{(1)} 
    \end{array}\right)\\
    &= 
    \left(\begin{array}{ccccccccc}
0 & 1 & 0 & 0 & 0 & 0 & 0 & 0 & 0 \\
 1 & 0 & 1 & 1 & 0 & 0 & 0 & 0 & 0 \\
 0 & 1 & 0 & 1 & 0 & 0 & 0 & 0 & 0 \\
 0 & 1 & 1 & 0 & 1 & 0 & 0 & 0 & 0 \\
 0 & 0 & 0 & 1 & 0 & 1 & 0 & 0 & 0 \\
 0 & 0 & 0 & 0 & 1 & 0 & 1 & 1 & 0 \\
 0 & 0 & 0 & 0 & 0 & 1 & 0 & 1 & 0 \\
 0 & 0 & 0 & 0 & 0 & 1 & 1 & 0 & 1 \\
 0 & 0 & 0 & 0 & 0 & 0 & 0 & 1 & 0
\end{array}\right)    \,,
\end{split}
\end{equation}
where the first level $(k=1)$ is made up of the $3\times3$ matrices on the main diagonal, and the second level $(k=2)$ is the full matrix. This simple example of the second generation of the network already illustrates the nested structure of the adjacency matrix.

Consider now for example the column with index $j=8$, which corresponds to $j'=j-1=7$, and which has the following representation in our base 3 notation: $\vec{\alpha}_8 = (\alpha_1,\alpha_2) = (1,2)$. From this, one can read off that the index $j=8$ is located in the right column of the $k=2$ level, which is the uppermost level, and in the middle column of the $k=1$ level, or unit cell. From the latter, the node degree $k_8$ gains a contribution of $1+\delta_{\alpha_1,1} = 2$. The remaining factor of one needed to obtain the exact result $k_8 = 3$ will be discussed now.

For levels $k\geq2$, the block matrices located on the sub- and superdiagonals are nonzero and are given by transposing the matrix $C^{(n)}$ either along the main diagonal, or along the antidiagonal. For $n=2$, the block matrices are given by
\begin{equation}
\label{eq:C_matrices}
\begin{split}
    C^{(1)} &= \begin{pmatrix}
        0 & 1 & 1\\
        0 & 0 & 0\\
        0 & 0 & 0 \\
    \end{pmatrix}\\
    (C^{(1)})^t &= \begin{pmatrix}
        0 & 0 & 0 \\
        1 & 0 & 0\\
        1 & 0 & 0\\
    \end{pmatrix}\\
    J^{(1)}C^{(1)}J^{(1)} &= \begin{pmatrix}
        0 & 0 & 0 \\
        0 & 0 & 0 \\
        1 & 1 & 0 \\
    \end{pmatrix}\\
    J^{(1)}(C^{(1)})^t J^{(1)} &= \begin{pmatrix}
        0 & 0 & 1\\ 
        0 & 0 & 1\\
        0 & 0 & 0\\
    \end{pmatrix}\,,
\end{split}
\end{equation}
where $J^{(n)}$ is the $3^n\times3^n$ exchange matrix, i.e., for $n=1$, this is
\begin{equation}
    \label{eq:exchange_matrix}
    J^{(1)} = \begin{pmatrix}
        0 & 0 & 1\\
        0 & 1 & 0\\
        1 & 0 & 0\\
    \end{pmatrix}\,.
\end{equation}
The first row in the $C^{(n)}$ matrix will be denoted as the vector $\vec{v}^{(n)} = (\vec{0},\vec{v}^{(n-1)},\vec{v}^{(n-1)})$, where $\vec{0} = (0,...,0)$ is the zero vector with $3^{n-1}$ elements. We list here the first three $\vec{v}^{(n)}$:
\begin{equation}
    \label{eq:v}
    \begin{split}
        \vec{v}^{(0)} &= (1)\\
        \vec{v}^{(1)} &= (0,1,1)\\
        \vec{v}^{(2)} &= (0,0,0,0,1,1,0,1,1)\\
    \end{split}
\end{equation}
From the structure of \eqref{eq:v}, it can be seen that for $n\geq2$ the first $(3^n -1)/2$ elements in the vector $\vec{v}^{(n)}$ are identically zero. The first two nonzero elements are then located at positions $j = (3^n +1)/2$ and $j = (3^n +3)/2$. Subsequent nonzero elements are separated by gaps of unequal size. 

With moderate effort, one can deduce that subsequent nonzero elements are located at indices that are shifted by a factor of $(3i + \sum_{m=1}^{i} 3^{a(m)})/2$ with respect to the positions $j = (3^n +1)/2$ and $j = (3^n +3)/2$. Herein, $a(m)$ is the so-called ruler function (OEIS A001511), the elements of which are the exponent of the largest power of $2$ which divides a given number $2m$ \cite{allouche_shallit_2003}. This sequence can be characterized by the following recurrence relations:
\begin{align}
\label{eq:ruler_recurrence}
        a(2m+1) &= 1 &&  a(2m) = 1+ a(m)\,.
\end{align}
A plot of the magnitude of the elements of the ruler function $a(m)$ is shown for $m = 1, 2, \dots,255$ in Fig. \ref{fig:ruler}. One can now easily see why it is called as such, since the rising and falling of the sequence as a function of the index resembles the markings on a classical ruler.

\begin{figure}[htp]
    \centering
    \includegraphics[width=0.75\linewidth]{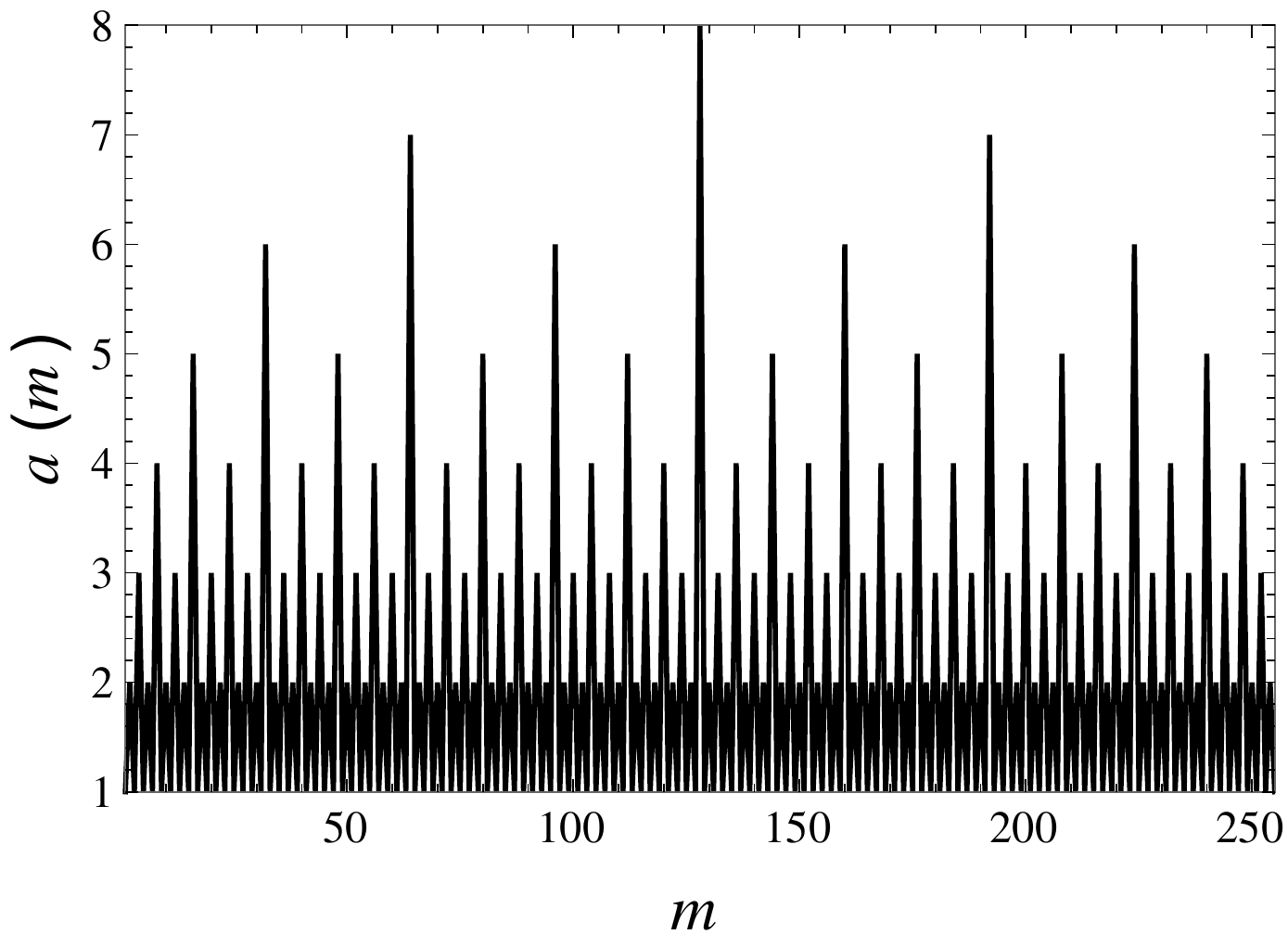}
    \caption{Graphical representation of the ruler sequence $a(m)$ generated by the recurrence relations \eqref{eq:ruler_recurrence} for $m = 1, 2, ... 255$. The pattern resembles the markings on a classical ruler.}
    \label{fig:ruler}
\end{figure}

After some simple algebra, the elements $v_j^{(n)}$, with $j=1,...,3^n$ can be fully determined by
\begin{equation}
    \label{eq:v_determined}
    v_j^{(n)} = \sum\limits_{i=0}^{2^{n-1}-1}\left(\delta_{j,\,\Xi_{ni}+\frac{1}{2}} + \delta_{j,\,\Xi_{ni}+\frac{3}{2}}\right)
\end{equation}
with
\begin{equation}
    \label{eq:Xi}
    \Xi_{ni} = \frac{1}{2}[3^n +3i +\sum\limits_{m=1}^i 3^{a(m)}]\,.
\end{equation}
The number of nonzero elements in $v^{(n)}$ can be deduced by noticing that in a generation $n+1$, this number has doubled with respect to generation $n$, and the first generation possesses two nonzero elements. Hence, the total number of nonzero elements is equal to $2^n$.

When the index $j$ is located in either the left or right column of the $k$th level (i.e., when $\alpha_k =0$ or $\alpha_k=2$, respectively), it picks up an extra factor of $+1$ for the node degree $k_j$ based on its position within the vectors $\vec{v}^{(k-1)}$ for $\alpha_k = 0$ or $\vec{v}^{(k-1)}\cdot J^{(k-1)}$ for $\alpha_k = 2$. Moreover, when $j$ is located in the middle column $(\alpha_k = 1)$ at positions $j = 3^{k-1}+1$ or $j =2\cdot3^{k-1}$, the node degree $k_j$ picks up an extra $2^{k-1}$.

Summing over levels $k$ and shifting the indices for simplicity, we arrive after some cumbersome algebra at the final result for the node degree $k_j(n)$, yielding the expression in equation \eqref{eq:degree}.

\section{Height distribution in the HRDM}\label{app:B}
We calculate the height probability density function (PDF) $\mathcal{P}^{(n)}(H = h)$ or, alternatively, the cumulative height distribution function (CDF) $\mathcal{F}^{n}(H\leq h)$ for the random hierarchical deposition model. Let us consider the height as a geometric random walk and let $\lambda^{-n}$ be the step size at time $n$, and $X_n$ a random variable that can assume the values $-\lambda^{-n}, 0$ or $\lambda^n$, according to whether the $n$th step is to the ``left'', ``middle'' or ``right'', which correspond to particle deposition, no deposition, or erosion, respectively. The random walk is then the countable sequence 
\begin{equation}
    \label{eq:random_walk_sequence}
    X_1,\,X_1 + X_2,\, X_1 + X_2 + X_3,\hdots
\end{equation}
and the height $H^{(n)}$ in generation $n\geq1$ is consequently
\begin{equation}
    \label{eq:height_definition}
    H^{(n)} = \sum_{j=1}^n X_j\,,
\end{equation}
which, for the geometric random walk with $\lambda\geq2$ equals
\begin{equation}
    H^{(n)} = \sum_{j=1}^n a_j \lambda^{-j}\,,\, a_j\in\{-1,0,1\}
\end{equation}
The distribution $\mathcal{P}$ of $H$ is a measure whose cumulative distribution function, defined as $\mathcal{F}(x) = \mathcal{P}([-1/2,x))$ is continuous and non-decreasing. Note that the measure is singular with respect to the Lebesgue measure, and hence we will from now on only work with the cumulative distribution function $\mathcal{F}$. We show that $\mathcal{F}$ can be found exactly for the hierarchical deposition model for $n\rightarrow\infty$, albeit in an integral form that must be solved numerically. 

\begin{theorem}
Consider the limiting ($n\rightarrow\infty$) hierarchical random deposition process where $H = \sum_{j=1}^\infty X_j$ and $X_j$ are independent random variables that have values $-\lambda^{-j}$, $0$ or $\lambda^j$ with probabilities $Q$, $1-P-Q$ and $P$, respectively. The cumulative distribution function $\mathcal{F}(h) = \mathcal{P}(H\leq h)$ is the following:
\begin{widetext}
\begin{equation}
    \label{eq:cumulative_full}
    \mathcal{F}(h) = \frac{1}{2\pi i}\int_\mathbb{R} \mathrm{d}t\,\frac{\me^{it/(\lambda-1)} - \me^{-ith}}{t}\prod_{j=1}^\infty\left[P \me^{it\lambda^{-j}} + Q
    \me^{-it\lambda^{-j}} + (1-P-Q)\right]
\end{equation}
\end{widetext}
\end{theorem}
\begin{proof}
The characteristic function $\varphi_H(t)$ for the height for $n\rightarrow\infty$ can be written as
\begin{equation}
    \label{eq:char_function}
    \begin{split}
        \varphi_H(t) &= E\left[\me^{itH}\right] = E\left[\me^{it\sum_j X_j}\right]\\
        &= \prod_{j=1}^\infty\left[P \me^{it\lambda^{-j}} + Q
    \me^{-it\lambda^{-j}} + (1-P-Q)\right]\,,
    \end{split}
\end{equation}
where in the final equality we have used the fact that the $X_j$ are independent random variables and the expected value can be factored. Furthermore, we can define the characteristic function as $\varphi_H(t) = \int_\mathbb{R} \me^{itx}\mathrm{d}\mathcal{F}(x)$. To proceed, we define the indicator function $\chi_h(x)$ as
\begin{equation}
    \label{eq:indicator}
    \chi_h(x) = \begin{cases}
        1\,,&\qquad if\, x\in(-\frac{1}{\lambda-1},h]\\
        0\,, &\qquad else\,,
    \end{cases}
\end{equation}
which can also be defined through its Fourier transform $\hat{\chi}_h(t)$, i.e.,
\begin{equation}
    \label{eq:indicator_fourier}
    \chi_h(x) = \int_\mathbb{R} \frac{\me^{i t x}}{2\pi i t}\left(\me^{it/(\lambda-1)} - \me^{-i t h}\right)\,\mathrm{d}t\,.
\end{equation}
The cumulative distribution function can now be written as
\begin{equation}
\begin{split}
    \mathcal{F}(h) &=\int_\mathbb{R} \chi_h(x)\mathrm{d}\mathcal{F}(x)\\
    &= \frac{1}{2\pi i}\int_\mathbb{R}\frac{\me^{it/(\lambda-1)} - \me^{-i t h}}{t}\varphi_H(t)\mathrm{d}t\,,
\end{split}
\end{equation}
which, upon inserting the expression for $\varphi_H(t)$ proves the theorem.
\end{proof}
The cumulative distribution function as defined by equation \eqref{eq:cumulative_full} cannot be calculated analytically as far as we know, but numerical approximations by truncating the infinite product can yield high-accuracy results. In Fig. \ref{fig:cantor}, we show the PDF and the CDF for the ternary Cantor set and a set with $P=1/3$ and $Q=1/4$, and compare results obtained by exact counting of all possible heights with the infinite-generation approximation in equation \eqref{eq:cumulative_full}. The results coincide quite well. 
\begin{figure}[htp]
    \centering
    \includegraphics[width=\linewidth]{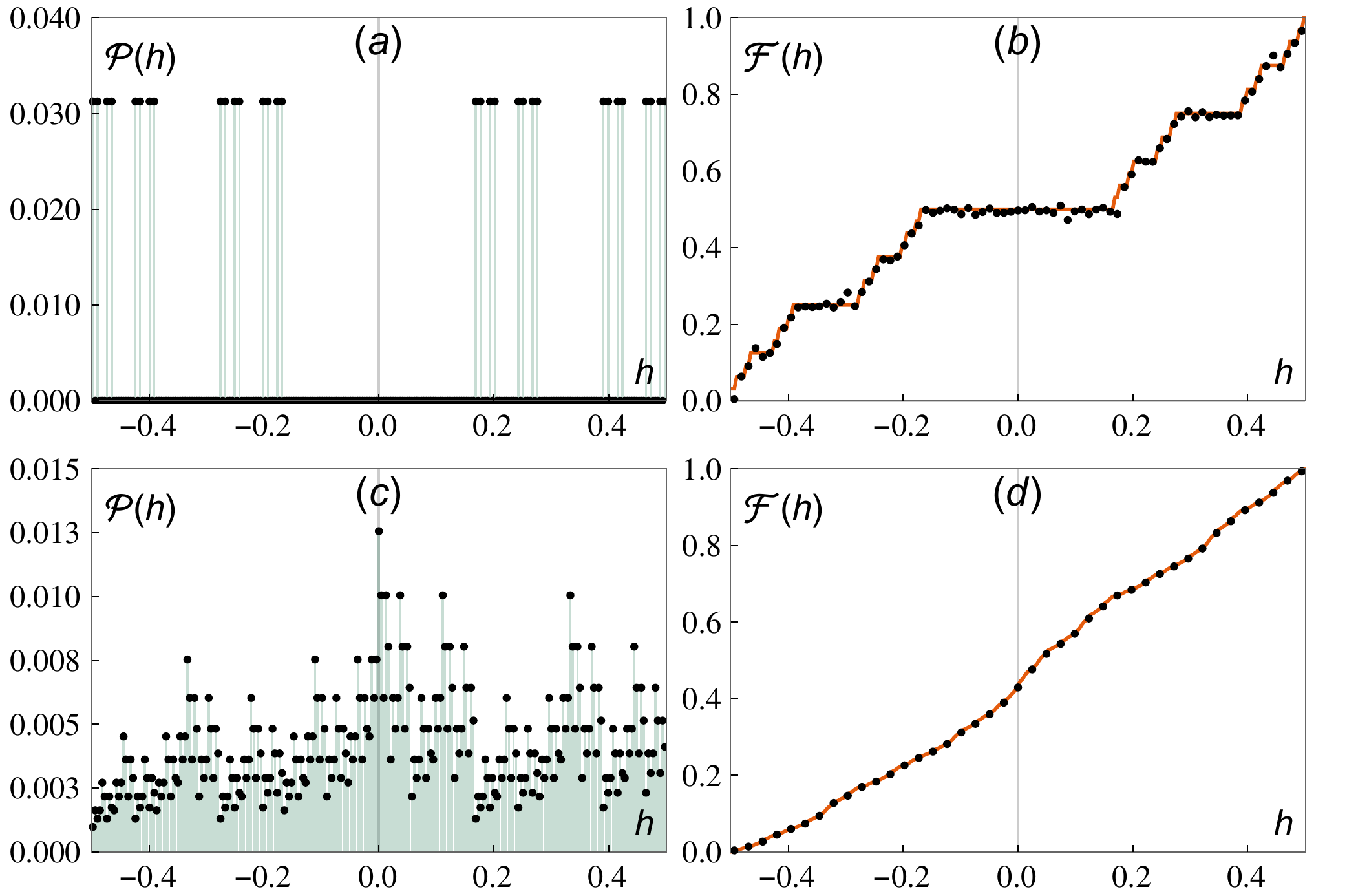}
    \caption{\textbf{(a)-(b)} Height PDF and CDF, respectively, for $P = Q = 1/2$, with the ternary Cantor set as support. \textbf{(c)-(d)} Height PDF and CDF, respectively, for $P = 1/3$ and $Q = 1/4$. Orange lines indicate exact distributions and black dots are numerical approximations obtained from Eq. \eqref{eq:cumulative_full}. Note that the same pattern repeats at multiple scales, indicating the fractal behaviour of the PDF. For both parameter choices, $n=5$ and $\lambda=3$ were used to find the exact PDF/CDF, and the product in equation \eqref{eq:cumulative_full} is truncated to $j = 100$.}
    \label{fig:cantor}
\end{figure}

Note that from the characteristic function all of the moments of $\mathcal{P}(h)$ can be calculated by differentiation, i.e., $\langle h^m\rangle = i^{-m} \varphi^{(m)}_H(0)$, where the superscript denotes $m$-fold differentiation. After some simple algebra, the mean and variance are then found to equal
\begin{align}
    \langle h\rangle &=\frac{P-Q}{\lambda-1}\\
    \sigma_h^2 &= \langle h^2\rangle - \langle h\rangle^2 = \frac{(P+Q) - (P-Q)^2}{\lambda^2 -1}\,.
\end{align}
For $\lambda=3$, $P=Q=1/2$, the mean and variance equal $\langle h\rangle = 0$ and $\sigma_h^2 = 1/8$, respectively. For $P=1/3$ and $Q=1/4$, they equal $\langle h\rangle = 1/24$ and $\sigma_h^2 = 83/1152 \approx 0.072$, respectively. This is confirmed by numerical calculation.

\section{Numerical degree exponents, moments, diameter and clustering for the HRDM}\label{app:C}
To numerically find the degree exponent $\gamma$ of the scale-free networks generated by the HVA, we solve the transcendental equation for the maximum likelihood estimator (MLE) $\hat{\gamma}$ \cite{Clauset2009}:
\begin{equation}
    \label{eq:MLE2}
    \frac{\zeta'(\hat{\gamma},k_{\text{min}})}{\zeta(\hat{\gamma},k_{\text{min}})} = -\frac{1}{m}\sum\limits_{i=1}^m \ln{k_i}\,,
\end{equation}
where $\zeta$ is the Hurwitz zeta function \cite{nist}, $m$ is the total number of observed values of the degrees $k$, and $k_{\text{min}}$ is the cutoff degree above which the distribution is a discrete power law. The prime denotes differentiation with respect to the first argument. The standard error $\sigma$ associated with the MLE is
\begin{equation}
    \label{eq:std_error}
    \sigma = \left[m \left(\frac{\zeta''(\hat{\gamma},k_{\text{min}})}{\zeta(\hat{\gamma},k_{\text{min}})}\right) -\left(\frac{\zeta'(\hat{\gamma},k_{\text{min}})}{\zeta(\hat{\gamma},k_{\text{min}})}\right)^2\right]^{-\frac{1}{2}}\,.
\end{equation}
Alternatively, the exponent can be found approximately, whereby the true power-law distributed integers are approximated as continuous reals rounded to the nearest integer, i.e.,
\begin{equation}
    \label{eq:exponent_approximate}
    \begin{split}
        \hat{\gamma}' &\simeq 1 + m \left[\sum\limits_{i=1}^m \ln\frac{k_i}{k_{\text{min}}-\frac{1}{2}}\right]^{-1}\,.
    \end{split}
\end{equation}
The error can be found in the same manner as before. We complement our analysis with a Kolmogorov-Smirnov (KS) test \cite{Clauset2009} on the cumulative degree distribution to find the optimal $\hat{\gamma}$ (or $\hat{\gamma}'$) and $k_{\text{min}}$ that fit the simulation data. We show the MLE results for the estimators $\hat{\gamma}$ and $\hat{\gamma}'$ in Table \ref{tbl:exponent}.

\begin{table}[htp]
\centering
\caption{Evolution of $\hat{\gamma}$ and $\hat{\gamma}'$ exponents with associated standard errors $\sigma$ and $\sigma'$. Parameters are $P = Q =0.25$ and $\lambda = 3$ and results are averaged over 5000 realizations.}
\begin{tabular}{|c|c|c|c|c|c|c|} 
\hline
\textbf{n}         & \textbf{3}     & \textbf{4}     & \textbf{5}     & \textbf{6}     & \textbf{7}     & \textbf{8}      \\ 
\hline
$\hat{\gamma}$  & 6.476 & 5.809 & 4.888 & 4.437 & 4.045 & 3.810  \\ 
\hline
$\sigma$  & 0.066 & 0.040 & 0.016 & 0.013 & 0.008 & 0.012  \\ 
\hline
$\hat{\gamma}'$ & 7.379 & 6.179 & 5.081 & 4.507 & 4.083 & 3.830  \\ 
\hline
$\sigma'$ & 0.017 & 0.008 & 0.004 & 0.004 & 0.003 & 0.006  \\
\hline
\end{tabular}
\label{tbl:exponent}
\end{table}

As a function of the total event probability $S = P+Q$, the mean clustering coefficient $\langle C\rangle$, the mean degree $\langle k\rangle$, the degree variance $\sigma_k^2$ and the degree exponent $\gamma$ are shown in Fig. \ref{fig:Splot}. The lines connecting data points are only for visualisation purposes.

\begin{figure}[htp]
    \centering
    \includegraphics[width=\linewidth]{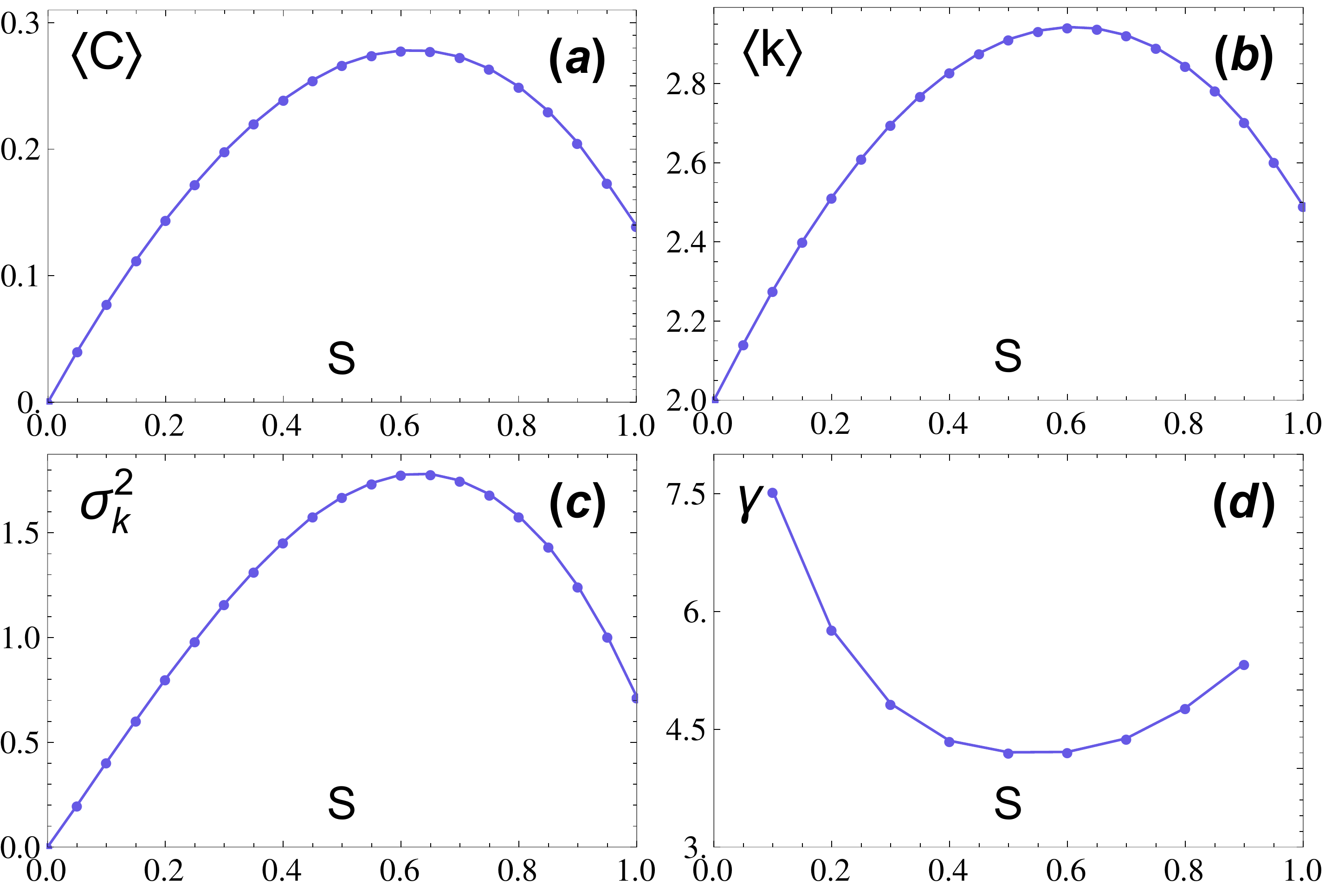}
    \caption{\textbf{(a) - (c)} Mean clustering coefficient, mean degree and degree variance as a function of $S=P+Q$ for $n=5$. \textbf{(d)} Degree exponent $\gamma$ as a function of $S=P+Q$ for $n=8$. Results are averaged over 5000 realisations.}
    \label{fig:Splot}
\end{figure}

The network diameter $D$ is shown in Fig. \ref{fig:diameter}(a) as a function of the number of vertices $N = 3^n$ for fixed deposition probability $P=0.2$. Full lines are power-law fits to the numerical data. The power-law behaviour of $D(N)\sim N^\epsilon$ is clear from Fig. \ref{fig:diameter}(b). Lines are only shown for visualisation purposes. 
\begin{figure}[htp]
    \centering
    \includegraphics[width=0.95\linewidth]{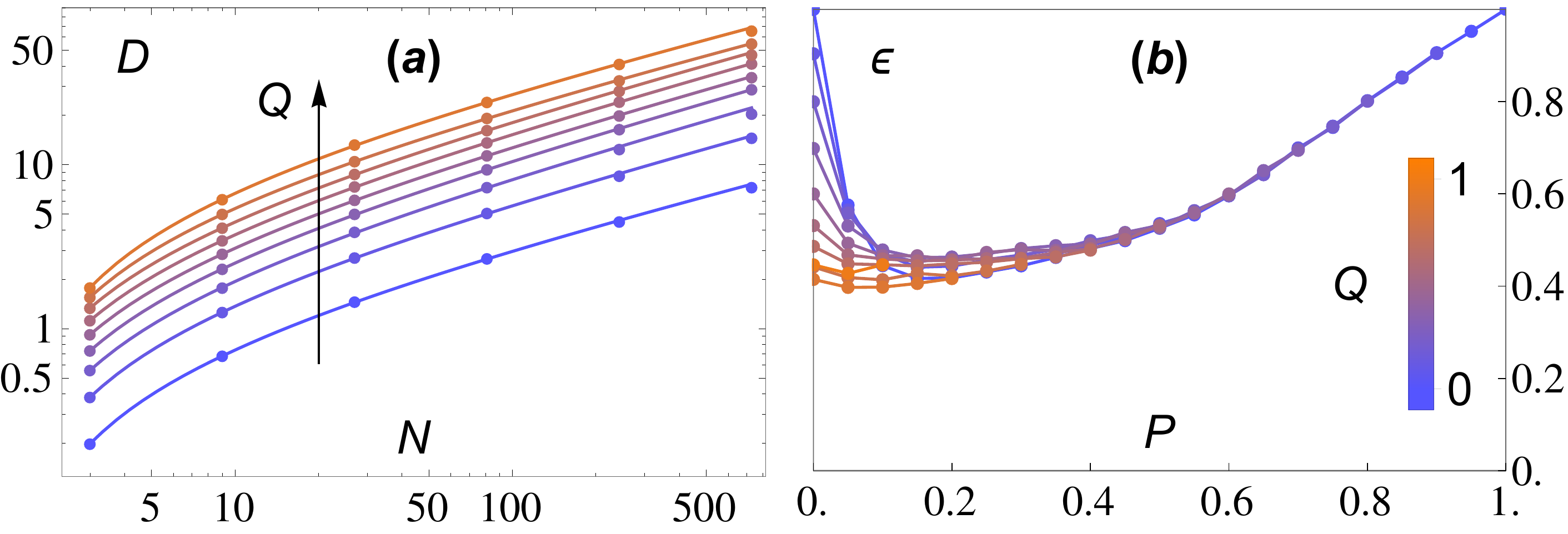}
    \caption{\textbf{(a)} The HVG diameter $D$ as a function of $N=3^n$ for fixed $P=0.2$ and $0\leq Q\leq P$ in steps of $\Delta Q =0.1$, and \textbf{(b)} the scaling exponent $\epsilon$ as a function $P$ for fixed $n=6$ and $0\leq Q\leq1$ n steps of $\Delta Q =0.1$. Results are averaged over 5000 realisations.}
    \label{fig:diameter}
\end{figure}

\bibliographystyle{apsrev4-2}
\bibliography{biblio.bib}

\end{document}